\documentclass[a4paper,fleqn]{cas-sc}

\usepackage[numbers]{natbib}

\def\tsc#1{\csdef{#1}{\textsc{\lowercase{#1}}\xspace}}
\tsc{WGM}
\tsc{QE}

\usepackage{amsmath,amssymb}  
\usepackage{mathrsfs}
\usepackage{algorithm, algorithmic}  
\usepackage{graphicx}   
\usepackage{amsthm} 
\usepackage{setspace}
\usepackage{subfigure}
\usepackage{enumerate}

\usepackage{verbatim}  
\usepackage{array} 
\usepackage{makecell} 
\usepackage{multirow}
\usepackage{url}  

\usepackage{hyperref}
\hypersetup{
colorlinks=true,
linkcolor=black,
anchorcolor=black,
filecolor=black,
urlcolor=black,
citecolor=black,}

\usepackage[capitalize]{cleveref} 

\def\C{{\mathbb{C}}}
\def\R{{\mathbb{R}}}
\def\Z{{\mathbb{Z}}}

\def\P{{\mathbf{P}}}
\def\U{{\mathbf{U}}}
\def\W{{\mathbf{W}}}
\def\x{{\mathbf{x}}}
\def\X{{\mathbf{X}}}
\def\L{{\mathbf{L}}}

\def\F{{\mathscr{F}}}

\def\BL{{\mathcal{BL}}}

\def\T{{\mathcal{T}}}
\def\G{{\mathcal{G}}}
\def\J{{\mathcal{J}}}
\def\E{{\mathcal{E}}}
\def\V{{\mathcal{V}}}
\def\S{{\mathcal{S}}}
\def\I{{\mathcal{I}}}

\newtheorem{theorem}{Theorem}
\newtheorem{definition}{Definition}
\newtheorem{lemma}{Lemma}
\newtheorem{corollary}{Corollary}

\begin{document}
\let\WriteBookmarks\relax
\def\floatpagepagefraction{1}
\def\textpagefraction{.001}
\let\printorcid\relax 

\shorttitle{Sampling Theory of Jointly Bandlimited Time-vertex Graph Signals}   

\shortauthors{H. Sheng et al.}

\title[mode = title]{Sampling Theory of Jointly Bandlimited Time-vertex Graph Signals}

\tnotemark[1]

\tnotetext[1]{This work was supported by the Singapore Ministry of Education Academic Research Fund Tier 2 grant MOE2018-T2-2-019 and A*STAR under its RIE2020 Advanced Manufacturing and Engineering (AME) Industry Alignment Fund – Pre Positioning (IAF-PP) (Grant No. A19D6a0053).}

\author[1]{Hang Sheng} 
\ead{20110720036@fudan.edu.cn} 

\author[1,2]{Hui Feng} 
\cormark[1] 
\ead{hfeng@fudan.edu.cn}

\author[1]{Junhao Yu}
\ead{17210720048@fudan.edu.cn}

\author[3]{Feng Ji}
\ead{jifeng@ntu.edu.sg}

\author[1,2]{Bo Hu}
\ead{bohu@fudan.edu.cn} 

\address[1]{School of Information Science and Technology, Fudan University, Shanghai 200433, China}
\address[2]{Shanghai Institute of Intelligent Electronics \& Systems, Shanghai 200433, China.}
\address[3]{School of Electrical and Electronic Engineering, Nanyang Technological University, 639798, Singapore}

\cortext[1]{Corresponding author} 

\begin{abstract}
Time-vertex graph signal (TVGS) models describe time-varying data with irregular structures. The bandlimitedness in the joint time-vertex Fourier spectral domain reflects smoothness in both temporal and graph topology. In this paper, we study the critical sampling of three types of TVGS including continuous-time signals, infinite-length sequences, and finite-length sequences in the time domain for each vertex on the graph. For a jointly bandlimited TVGS, we prove a lower bound on sampling density or sampling ratio, which depends on the measure of the spectral support in the joint time-vertex Fourier spectral domain. We also provide a lower bound on the sampling density or sampling ratio of each vertex on sampling sets for perfect recovery. To demonstrate that critical sampling is achievable, we propose the sampling and reconstruction procedures for the different types of TVGS. Finally, we show how the proposed sampling schemes can be applied to numerical as well as real datasets.
\end{abstract}

\begin{highlights}
    \item We prove the necessary conditions for the stable reconstruction of JBL CTVGS. We prove a lower bound on total sampling density, which gives rise to the concept of \emph{critical sampling} for JBL CTVGS. We prove lower bounds on the sampling densities of the signals on subsets of vertices to be sampled. 
    
    \item We construct a multi-band sampling scheme for JBL CTVGS to prove that critical sampling is achievable for any JBL CTVGS. 

    \item We apply the sampling theory and multi-band sampling scheme to obtain critical sampling sets for JBL DTVGS and FTVGS.
\end{highlights}


\maketitle

\section{Introduction}
\label{sec:intr}

The ability of graphs to capture the underlying structure of data has been the driving force behind the utilization of graph signal processing (GSP) theory for addressing high-dimensional data on irregular domains. GSP extends traditional signal processing techniques to irregular data \cite{big2014}, including the graph Fourier transform (GFT) \cite{Spectral1997,emerging2013,big2014}, graph filtering \cite{filter2022}, graph sampling \cite{theory,sampling2022}, and graph signal estimation \cite{RN11,xie2019bayesian,estimation2022}. The introduction of these concepts and tools has facilitated the application of GSP theory in various practical domains, encompassing sensor networks \cite{RN8,sensor2022}, brain network analysis \cite{brain2022,brain2023}, and graph neural networks \cite{GNN2023}. The overview articles \cite{overview,review2020} contain comprehensive discussions on GSP and its applications.

In numerous scenarios involving sensor, social, or biological networks, graph signals always vary with time. Consequently, time-vertex graph signal (TVGS) processing models have been proposed. Grassi \emph{et al.}\cite{timevertex} stacks graph signals of multiple moments referred to as finite time-vertex graph sequences (FTVGS). Later, signal spaces are modeled using general Hilbert spaces in \cite{ji2019hilbert}, allowing the collective temporal signals to represent both discrete time-vertex graph sequences (DTVGS) and continuous time-vertex graph signals (CTVGS). The signal models FTVGS, DTVGS, and CTVGS are collectively referred to as TVGS in this paper. 

Given the combined characteristics of the time and vertex domains, a TVGS typically encompasses a substantial amount of data, which proves valuable for learning and analysis. However, we encounter challenges related to the observation, storage, and processing of the exponentially increasing data volume. Therefore, it is highly advantageous to sample partial observations instead of recording all the raw data, while still preserving the majority of the information of the signal. For example, an event camera\cite{DVS20,DVS2024} only captures changing pixels to offer better storage performance and high-speed response. In a social network, understanding the opinions of the total population by investigating a small number of candidates may save a lot of resources. 

When considering a single vertex in a generated trivial graph, the Nyquist-Shannon sampling theorem states that any bandlimited TVGS can be perfectly reconstructed if sufficiently many samples (at the Nyquist rate) are taken \cite{nyquist, shannon}. Alternatively, when focusing on a specific moment in a TVGS, the TVGS is a static graph signal. A sampling theory of a bandlimited graph signal defined by GFT is proposed in \cite{theory,2016eff}. If the signal on each vertex of TVGS is bandlimited and there are correlations among vertices, the TVGS shows smoothness in both the temporal and vertex domains. However, the transformation in a single domain cannot always reflect correlations in both the time and vertex domains. As we shall illustrate in \cref{fig:exp2_signal}, the joint time-vertex Fourier transform (JFT)\cite{2016JFT} offers a more compact spectral representation of a TVGS.

A TVGS bandlimited in joint time-vertex domain is called a \emph{jointly bandlimited (JBL) TVGS}. We can define the projection bandwidths for a JBL TVGS in time and vertex domains, respectively (See detail in \cref{df:probth_c,df:probth_f}). Therefore, a \emph{separate sampling} scheme has been proposed for sampling JBL FTVGS\cite{sampling2018}, which has also been extended to JBL CTVGS and DTVGS\cite{ji2019hilbert}. The separate sampling scheme achieves a sampling density or ratio which is the product of the two projection bandwidths. However, existing literature has not proved a lower bound on the sampling density of the JBL TVGS. Moreover, there is no method provided for sampling at the lowest sampling density. These two aspects are the focal points of this paper.

In the preliminary version\cite{Yu}, we introduce a concept of joint bandwidth for JBL FTVGS and prove that the number of samples required for stable recovery can be reduced when using joint bandwidth compared to separate sampling. We propose a scheme to sample JFT with minimum samples. 

In this paper, the sampling theory of three kinds of TVGS (CTVGS, DTVGS, and FTVGS) is provided. Following the outline of a signal processing textbook \cite{DSP}, we delve into the sampling of TVGS, commencing with CTVGS and progressing to discrete signals. We want to know if there is a lower bound on the sampling density for CTVGS to ensure stable reconstruction first. We expand the concept of joint bandwidth to CTVGS and establish a connection between CTVGS sampling and multiple-input multiple-output (MIMO) sampling \cite{MIMO}. Subsequently, we prove lower bounds on the sampling densities, wherein the total sampling density is the joint bandwidth. This sampling theory can apply to DTVGS and FTVGS. The exploration of the lower bound on the sampling ratio of FTVGS goes beyond the scope of \cite{Yu}, and we provide additional proofs not covered in \cite{Yu}. 

In short, the total sampling density of a JBL TVGS is lower-bounded by the joint bandwidth, which is smaller than (and in special cases equal to) the product of two projection bandwidths. So we can sample a JBL TVGS at the lowest density or ratio, which is no more than that of separate sampling.

Besides, once the vertices to be sampled are selected, we may have specific environmental or hardware constraints that require the sampling density or ratio of the signals on a subset of vertices to be as low as possible. Our proposed sampling theorems give a lower bound on the sampling density or ratio of signals on the subset of vertices to be sampled.

In summary, the main contributions of this paper include:
\begin{itemize}
    \item We prove the necessary conditions for the stable reconstruction of JBL CTVGS. We prove a lower bound on total sampling density, which gives rise to the concept of \emph{critical sampling} for JBL CTVGS. We prove lower bounds on the sampling densities of the signals on subsets of vertices to be sampled. 
    
    \item We construct a multi-band sampling scheme for JBL CTVGS to prove that critical sampling is achievable for any JBL CTVGS. 

    \item We apply the sampling theory and multi-band sampling scheme to obtain critical sampling sets for JBL DTVGS and FTVGS.
\end{itemize}

The rest of the paper is organized as follows. In \cref{sec:model}, we present the TVGS models and describe the JFT. In \cref{sec:cri_samp_c,sec:cri_samp_d,sec:cri_samp_f}, we discuss the necessary conditions for the stable sampling of CTVGS, DTVGS, and FTVGS respectively, and design the sampling procedures to achieve critical sampling. We provide numerical results in \cref{sec:exp} and conclude in \cref{sec:conclusion}. Notations are listed in \cref{tab:notation}.

\begin{table}[htbp]
\normalsize
	\centering
	\caption{NOTATIONS}
	\begin{tabular}{l l}  
			\hline 
			Notation                & Description         \\
			\hline
			$\G$        & an undirected graph  \\
			$\T$        & the topology in time domain \\
			$\V$        & a set of vertices  \\
            $\times$    & the Cartesian product \\
			$\L$        & the graph Laplacian matrix  \\
			$\U_\G$      & the eigenmatrix of $\L_\G$ in vertex domain\\
			$\U_\T$      & the eigenmatrix of $\L_\T$ in time domain \\
			$\otimes$   & the Kronecker product  \\
			$\U_\J$      & $\U_\G \otimes \U_\T$  \\
			$\X$        & a TVGS  \\  
			$\hat{\X}$   & a TVGS recovered from samples \\
            $\text{vec}(\cdot)$ & matrix vectorization \\
			$\x$        & $\text{vec}(\X)$   \\
			$\x_t$      & a graph signal at instant $t$ \\
			$\F_{FT}(\cdot)$    & the FT on each vertex of a CTVGS \\
            $\F_{DT}(\cdot)$    & the DTFT on each vertex of a DTVGS \\
            $\F_\T(\cdot)$    & the DFT (or GFT) on each vertex of an FTVGS with directed (or undirected) graph $\T$ \\
            $\F_\G(\cdot)$    & the GFT on each instant of a TVGS \\
            $\F_\J(\cdot)$    & the JFT on a TVGS \\
            $T_s$       & the sampling period of discrete time sequences \\
            $\mathcal{F}_i$ & the set of spectral support of $\F_\J(\X)(i, \cdot)$ \\
            $\I_f$       & an index set of nonzero elements in $\F_\J(\X)$ \\
            $\mu(\cdot)$ & the Lebesgue measure of a set \\
            $B$         & the joint bandwidth of a JBL TVGS  \\
            $B_\G$       & the projection bandwidth in the vertex domain \\
            $B_\T$       & the projection bandwidth in the time domain \\
            $\S_v$        & a discrete subset of $\T$  \\
            $\S_\G$     & $\{ v: (v, t_{vz}) \in \S \} $ \\
            $\S_\T$     & $\cup_{v \in \S_\G} \S_v$ \\
            $\S$        & the sampling set of $\X$\\
            $\S_\Theta$ & $\{ (v, t_{vz}): v \in \Theta \} \subseteq \S $ \\
            $\S_\G'$     & a subset of $\V_\G$ with $|\S_\G'| = B_\G$ \\
            $\S_\T'$     & a subset of $\T$ obtaining based on $B_\T$ \\
            $|\cdot|$      & the cardinality of a set \\
            $\mathbf{\Psi}$ & a sampling matrix  \\
            $\mathbf{\Phi}$ & an interpolation matrix  \\
            $D(\cdot)$  & the density of a set \\
            $R(\cdot)$  & the ratio of a set \\
            $\P$        & the projection operator \\
			\hline
	\end{tabular}
	\label{tab:notation}
\end{table}

\section{Models} \label{sec:model}

\subsection{Continuous time-vertex graph signals}
\label{subsec:mod_c}

Consider an undirected graph $\G = (\V_\G, \E_\G, \W_\G)$, where $\V_\G=\{ v_1, \dots, v_N \}$ is the set of vertices, and $\E_\G$ is the set of edges.
The matrix $\W_\G=[w(v_i, v_j)]_{v_i,v_j\in\V_\G}$ represents an $N \times N$ symmetric weighted adjacency matrix, where $w(v_i,v_j)$ is the weight of the edge between $v_i$ and $v_j$. If every vertex in $\G$ is associated with an $L^2$ function in time, such signals are called CTVGS. Let $\T = \R$ denote the topological space of CTVGS in the time domain. A CTVGS $\X = [\X(v_1, t), \X(v_2, t), \dots, \X(v_N, t)]^T \in L^2(\V_\G \times \T)$ is illustrated in \cref{fig:TVGS} (a) and (b), where $\X(\cdot,t), t \in \T$ is a graph signal at instant $t$, and $\X(v,\cdot), v \in \V_\G$ is a complex-valued $L^2$ function on vertex $v$.

We apply the Fourier transform (FT) independently to each $\X(v,\cdot), v \in \V_\G$ to get the spectrum 
\begin{equation} 
\label{eq:X_Tf}
    \F_{FT}(\X) (v,\Omega) = \int_{-\infty}^{+\infty} \X(v,t) e^{-j\Omega t} dt, v \in \V_\G, \Omega \in \R,
\end{equation}
where $\F_{FT}(\X)(v,\cdot)$ is an $L^2$ function. For example, \cref{fig:TVGS} (c) shows the spectral of each $\F_{FT}(\X)(v,\cdot)$ of a CTVGS with $4$ vertices. 

GFT is commonly used to obtain the spectrum of graph signals \cite{emerging2013,big2014} and captures the correlation among vertices. The degree matrix is the diagonal matrix $\mathbf{D}_\G=\text{diag}(d_i)$, where $d_i=\sum_j w(v_i, v_j)$ and the graph Laplacian matrix is defined by $\L_\G =\mathbf{D}_\G-\W_\G$. The matrix $\L_\G$ can be decomposed as $ \L_\G =\U_\G \mathbf{\Lambda}_\G \U_\G^H $, where the eigenvectors $\{ \mathbf{u}_i \}_{i=1}^N$ of $\L_\G$ form the columns of $\U_\G$, $\mathbf{\Lambda}_\G$ is a diagonal matrix of eigenvalues $\{ \lambda_i \} _{i=1}^N$ corresponding to $\{ \mathbf{u}_i \}$ \cite{overview}, and $\U_\G^H$ is the Hermitian of $\U_\G$. 
We apply GFT on $\X$ to get the spectrum 
\begin{equation}
    \label{eq:X_Gf_c}
    \F_\G(\X)(i,t) =\sum_{v \in \V_\G} \mathbf{u}_i(v) \X(v,t), i=1,\dots,N, t \in \R,
\end{equation}
where $\F_\G(\X)(i,\cdot)$ is a continuous-time signal corresponding to $\lambda_i$, and $\mathbf{u}_i(v)$ is the $v$-th component of $\mathbf{u}_i$. 

JFT is constructed by applying GFT and FT jointly \cite{2016JFT}. The JFT of CTVGS is defined as $\F_\J(\X)$ expressed as
\begin{equation} 
\label{eq:JFT_c}
    \F_\J(\X)(i, \Omega) = \F_\G(\F_{FT}(\X)) = \sum_{v \in \V_\G} \! (\mathbf{u}_i(v) \int_{-\infty}^{+\infty} \X(v, t) e^{-j\Omega t} dt )\!, i=1,\dots ,N, \Omega \in \R,
\end{equation}
where $\F_\J(\X)$ is a size $N$ vector of $L^2$ functions on $\Omega$. For example, \cref{fig:TVGS} (d) illustrates the spectrum of $\F_\J(\X)$. JFT portrays the frequency of each $\X(v,\cdot)$ by the FT and decomposes each $\F_{FT}(\X)(v,\cdot)$ according to different levels of smoothness between vertices (usually similarity) by GFT.

Correspondingly, the inverse joint time-vertex Fourier transform (IJFT) is
\begin{equation} 
\label{eq:IJFT_c}
    \X(v,t) = \F_\J^{-1}(\F_\J(\X)) = \F_\G^{-1}(\F_{FT}^{-1}(\F_\J(\X))) = \sum^N_{i=1} \! \left( \mathbf{u}_i(v) \frac{1}{2\pi}  \int_{-\infty}^{+\infty} \F_\J(\X)(i, \Omega) e^{j\Omega t} d\Omega \right), v\! \in \! \V_\G, t\! \in \! \R.
\end{equation}

\subsection{Discrete time-vertex graph sequences}
\label{subsec:mod_d}

In signal processing, it is often necessary to discretize continuous signals for practical applications. Consequently, we work with discrete TVGS, where the time domain is represented by $\T = \Z$. Each vertex in $\G$ (an undirected graph) is associated with an infinite $\ell^2$ discrete sequence. Such a signal is called a DTVGS in $\ell^2(\V_\G \times \T)$, as shown in \cref{fig:TVGS} (e) and (f). For a DTVGS $\X(v,n)$, $\X(\cdot,n), n \in \T$ is a graph signal at instant $nT_s$ for the sampling period $T_s$, and $\X(v,\cdot), v \in \V_\G$ is a complex-valued sequence on vertex $v$. In essence, a DTVGS can be regarded as a clock-synchronized discretization of a CTVGS. 

Although each $\X(v,\cdot)$ is a discrete sequence, the underlying connections between different vertices remain the same and still capture signal correlations. It is feasible to utilize transformations in the joint domain to handle DTVGS. The GFT remains unchanged, while in the time domain, the FT is replaced by the discrete-time Fourier transform (DTFT). 
\begin{equation} 
\label{eq:X_Tf_d}
    \F_{DT}(\X) (v,\omega) := \sum_{n \in \T} \X(v,n) e^{-j\omega n}, v \in \V_\G, \omega \in \R,
\end{equation}
where $\F_{DT}(\X)(v,\cdot)$ is a continuous periodic spectral function, shown in \cref{fig:TVGS}(g), and $\omega = \Omega T_s = \frac{\Omega}{f_s}$ is the digital angular frequency.

To form GFT, the spectrum of a DTVGS is  
\begin{equation} 
\label{eq:X_Gf_d}
    \F_\G(\X)(i,n) =  \sum_{v \in \V_\G} \mathbf{u}_i(v) \X(v,n), i=1, \dots, N, n \in \T ,
\end{equation}
where $\F_\G(\X)(i,\cdot)$ is a discrete time sequence.

Analogous to CTVGS, the JFT $\F_\J(\X)$ of a DTVGS $\X$ is defined as follows 
\begin{equation}
    \F_\J(\X)(i, \omega) = \F_\G( \F_{DT}(\X)) = \sum_{v \in \V_\G}  \mathbf{u}_i(v) \sum_{n \in \T} \X(v,n) e^{-j\omega n} , i=1,\dots ,N, \omega \in \R,
\end{equation}
where $\F_\J(\X)$ contains $N$ periodic $L^2$ functions with period of $2\pi$ on $\omega$, shown in \cref{fig:TVGS}(h). 

Similarly, the inverse transform IJFT of DTVGS is 
\begin{equation} 
\label{eq:IJFT_d}
    \X(v,n) = \F_\J^{-1}(\F_\J(\X)) = \F_\G^{-1}(\F_{DT}^{-1}(\F_\J(\X))) = \sum^N_{i=1} \left( \mathbf{u}_i(v) \frac{1}{2\pi} \int_{-\pi}^{\pi} \F_\J(\X)(i, \omega) e^{j\omega n} d\omega \right), v \! \in \! \V_\G, n \! \in \! \T.
\end{equation}

\begin{figure*} [htbp] 
	\centering	
    \includegraphics[scale=0.39]{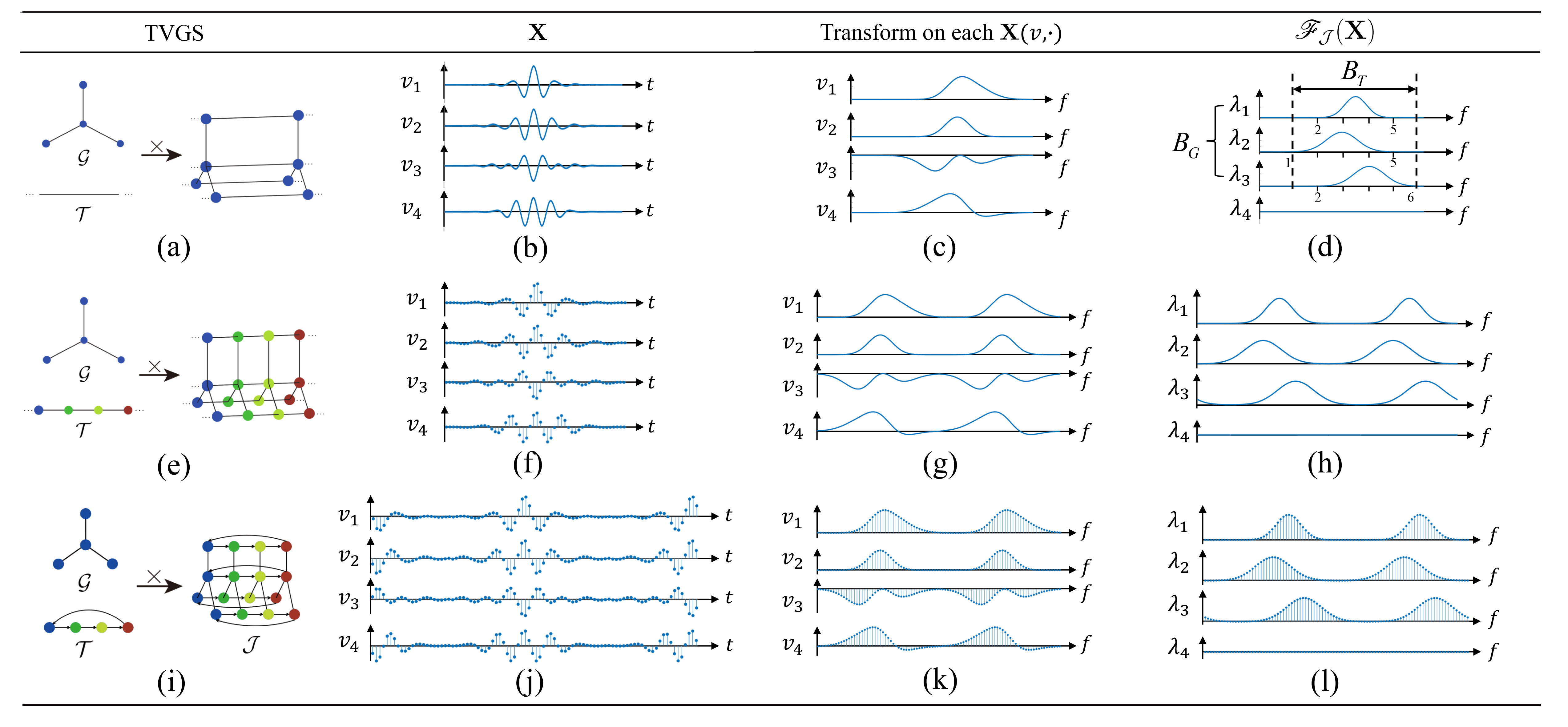}
	\caption{(a) is a CTVGS. (b) shows the signals on vertices of (a) along the temporal axis. (c) and (d) are the FT and JFT spectra of the signals in (b), respectively. (e) is a DTVGS. (f) shows the signals on vertices of (e) along the temporal axis, which discretizes the signals in (b). (g) and (h) are the DTFT and JFT spectra of the signals (f), respectively. (i) is an FTVGS, which is a period version of the signals in (f). (j) shows the signals on vertices of (i) along the temporal axis. (k) and (l) are the DFT and JFT spectra of the signals in (j), respectively. (Both signal and spectrum can be complex, and we only plot the real part in the figure.)} 
    \label{fig:TVGS}
\end{figure*}

\subsection{Finite time-vertex graph sequences}
\label{subsec:mod_f}

In the case of a DTVGS, if the sequence on each vertex is $T$-periodic as shown in \cref{fig:TVGS} (j), it suffices to analyze only a finite number of time stamps within one period. Such a signal is referred to as an FTVGS. If the periods $T_i$ of the signals on different vertices $v_i \in \V_\G$ are different, we chose $T$ to be the least common multiple of the periods $T_i$.

Let $\x_t \in \C^N$ denote the graph signal at instant $t$, and the FTVGS is represented by the matrix $\X=[\x_1, \x_2, \dots, \x_T] \in \C^{N\times T}$. Each row of $\X$ is a finite time sequence on a vertex, and each column is a static graph signal. Moreover, $\X$ can be vectorized using the $\text{vec}(\cdot)$ operator to give $\x = \text{vec}(\X) \in \C^{NT}$, which is a vector form of $\X$ by stacking its rows.

The domain of a finite periodic discrete time series can be represented by a directed cyclic graph $\T = (\V_\T, \E_\T, \W_\T)$ \cite{Moura,Moura2}, shown in \cref{fig:TVGS} (i), where $\V_\T=\{ n_1, \dots, n_T \}$ is the set of vertices, $\E_\T$ is the set of directed edges, $\W_\T$ is the asymmetric weighted adjacency matrix. The eigenvectors obtained from the Jordan decomposition of $\W_\T$ serve as the bases for the discrete Fourier transform (DFT) \cite{timevertex}. So we apply the DFT to each row of $\X$
\begin{equation*} 
    \F_\T(\X) := \X \overline{\U}_\T,
\end{equation*}
where $\overline{\U}_\T$ is the complex conjugate of $\U_\T$, and $\U_\T \in \C^{T\times T}$ is the normalized DFT bases
\begin{equation*}
    \U^H_\T (n,m) = \frac{1}{\sqrt{T}} e^{-\frac{j2\pi (m-1) n}{T}}, n,m = 1,2,\dots,T.
\end{equation*}
$\F_\T(\X)$ is a matrix representing the values within one period in \cref{fig:TVGS} (k). On the other hand, it is also possible to model the time domain by an undirected cycle graph $\T = (\V_\T, \E_\T, \W_\T)$ \cite{sampling2018}. The graph Laplacian $\L_\T$ can be decomposed as $\L_\T = \U_\T \mathbf{\Lambda}_\T \U_\T^H$. Then,  we can apply GFT to each row of $\X$, and $\F_\T(\X) := \X \overline{\U}_\T$ remains valid. Whether the temporal topology is represented as a directed or undirected graph does not affect the subsequent theoretical results.

In the vertex domain, we apply GFT to each column of $\X$,
\begin{equation*} 
    \F_\G(\X) := \U_\G^H \X .
\end{equation*}
The graphical model of FTVGS, denoted by $\J$ (stands for ``joint''), is the Cartesian product of $\G$ and $\T$ (illustrated in \cref{fig:TVGS} (i)), 
\begin{equation*}
    \J = \G \times \T = (\V_\G \times \V_\T, \E_\J).
\end{equation*}
The product structure is proposed to take full advantage of the intrinsic structure of FTVGS, which alleviates the issue of the curse of dimensionality \cite{big2014,sampling2018}.

Then we can denote the Laplacian matrix of $\J$ as $\L_\J$ \cite{timevertex}, 
\begin{equation*}
    \L_\J = \L_\G \times \L_\T = (\U_\G \otimes \U_\T)(\mathbf{\Lambda}_\G \times \mathbf{\Lambda}_\T)(\U_\G \otimes \U_\T)^H = \U_\J \mathbf{\Lambda}_\J \U_\J^H .
\end{equation*}
JFT has been introduced by applying the transform on $\T$ in time domain and $\G$ in vertex domain, respectively \cite{timevertex}
\begin{equation}
\label{eq:JFT_f}
    \F_\J(\X) := \U_\G^H \X \overline{\U}_\T ,
\end{equation}
shown in \cref{fig:TVGS}(l). Note that $\F_\J(\X)$ is a matrix representing the values within one period in \cref{fig:TVGS}(l). In the vector form, the transformation can be expressed as
\begin{equation} 
\label{F_xf}
    \F_\J(\x) := \U_\J^H \x ,
\end{equation}
where $\U_\J = \U_\G \otimes \U_\T$, $\F_\J(\x) \in \R^{NT}$ is a vector.

The inverse transform IJFT of $\X$ is 
\begin{equation*}
    \X = \F_\J^{-1}(\F_\J(\X)) = \U_\G \F_\J(\X) \U_\T^T,
\end{equation*}
where $\U_\T^T$ is the transpose of $\U_\T$, and its vectorized version satisfies $ \x =\F_\J^{-1}(\F_\J(\x)) = \U_\J \F_\J(\x)$.

It is important to emphasize that although the graph domain topology of the three kinds of TVGS in this paper is modeled as an undirected graph $\G$, \emph{our theory still holds when modeling the graph domain topology as a directed graph}. In other words, the graph domain topology of a TVGS can be modeled as either a directed or undirected graph depending on the specific application.

\section{Problem Formulation}
\label{sec:prob}

\subsection{Sampling for CTVGS}
For a CTVGS $\X$, suppose that a function $\X(v,\cdot) \in L^2(\R)$ is sampled on a discrete set $\S_v=\{ t_{vz}:z \in \Z \} \subseteq \T$. The sampling on vertex $v$ can be expressed as
\begin{equation*}
    \sum_{z\in\Z} \X(v,t) \delta(t-t_{vz}),
\end{equation*}
by which the sampled signal on $\S_v$ is $\X(v,t \in \S_v)$, denoted as $\X(v,\S_v) $ for brevity and the same below.

The sampling set of $\X$ is denoted by $\S=\{ (v, t_{vz}): v \in \V_\G, t_{vz} \in \T \}$. We define the projected sampling set of $\S$ in vertex domain as $\S_\G=\{ v: (v, t_{vz}) \in \S \}$ and the projected sampling set of $\S$ in the time domain as 
\begin{equation*}
    \S_\T = \bigcup\limits_{v \in \S_\G} \S_v = \{ t_{vz}: (v, t_{vz}) \in \S \}.
\end{equation*}
The sampling of $\X$ in the vertex domain can be described as $\X(\S_\G, \cdot)$ such that
\begin{equation*}
    \X(\S_\G, t) := \left[ \begin{matrix} \X(v_{s_1}, t) \\ \X(v_{s_2}, t) \\ \vdots \\ \X(v_{s_{|\S_\G|}}, t) \end{matrix} \right] = \mathbf{\Psi}_\G \left[ \begin{matrix} \X(v_1, t) \\ \X(v_2, t) \\ \vdots \\ \X(v_N, t) \end{matrix} \right],
\end{equation*}
where $\mathbf{\Psi}_\G=[\psi (p,q)]\in \{ 0, 1 \}^{|\S_\G| \times N}$ is the sampling matrix corresponding to $\S_\G=\{ v_{s_1}, v_{s_2}, \dots, v_{s_{|\S_\G|}} \}$, defined as 
\begin{equation*}
    \psi (p,q) = \left \{ 
        \begin{array}{ll}
        1, & {q = v_i;} \\ 
        0, & {\text{otherwise}.} \end{array}\right.
\end{equation*}

The sampling set $\S_v$ can be either uniform or non-uniform in the time domain. To deal with non-uniform sampling sets, it is useful to define \emph{sampling densities} as
\begin{equation*}
    D(\S_v) := \mathop{{\rm lim\ inf}}\limits_{t \to \infty} \frac{|\S_v \cap [-t,t] |}{2t},
\end{equation*}
which is equivalent to the definition of sampling density in \cite{MIMO}. The \emph{total sampling density} of $\X$ can be expressed as
\begin{equation*}
    D(\S) := \mathop{{\rm lim\ inf}}\limits_{t \to \infty} \frac{|\S \cap \{ \V_\G \times [-t,t]\} |}{2t N},
\end{equation*}
where $N = |\V_\G|$.

For CTVGS, our goal is to develop a sampling theorem and give a feasible sampling and reconstruction procedure (\cref{alg:multi}). The sampling theorem for CTVGS includes 
the selection of $\S_\G$ (\cref{cor:rank_k}) and lower bounds on $D(\S_\Theta), \S_\Theta \subseteq \S$ (\cref{thm:subset_c}). The importance of proving the sampling theory is that 
\begin{enumerate}[(i)]
    \item Theoretically, we give the necessary conditions for reconstructing CTVGS from samples and prove the lower bounds of $D(\S_\Theta)$. Additionally, the sampling theory of CTVGS is the basis of the sampling theorems of DTVGS and FTVGS (see more detail in proof of \cref{thm:subset_d,thm:subset_f}).

    \item In practice, the sampling of CTVGS is not feasible in a discrete-time operating computer system. It does not mean that the sampling theorem and scheme are not valuable. Real signals are generally continuous-time signals, and the sampling theorem and scheme of CTVGS can guide us on how to arrange the location of the sensors, how to design the filters, and at what frequency each sensor should sample.
\end{enumerate}

\subsection{Sampling for DTVGS}

For a DTVGS $\X$, suppose $\S_v=\{ n_{vz}:z \in \Z \} \subseteq \T$ is the sampling set of $\X(v,\cdot) \in \ell^2(\Z)$. The sampled signal on $\S_v$ is $\X(v,n \in \S_v)$, denoted as $\X(v,\S_v)$ for brevity.

Analogously, the sampling set of $\X$ is $\S=\{ (v, n_{vz}): v \in \V_\G, n_{vz} \in \T \}$. We define the projected sampling set of $\S$ in vertex domain as $\S_\G=\{ v: (v, n_{vz}) \in \S \}$, and the projected sampling set of $\S$ in time domain as 
\begin{equation*}
    \S_\T = \bigcup\limits_{v \in \S_\G} \S_v = \{ n_{vz}: (v, n_{vz}) \in \S \}.
\end{equation*}
The sampling of $\X$ in the vertex domain is the same as that of CTVGS, expressed as $\X(\S_\G,\cdot) = \mathbf{\Psi}_\G \X$.

The sampling of discrete sequences is measured by the \emph{sampling (downsampling) ratios} defined as
\begin{equation*}
    R_D(\S_v) := \mathop{{\rm lim\ inf}}\limits_{n \to \infty} \frac{|\S_v \cap [-n, n] |}{2n}.
\end{equation*}
The \emph{total sampling ratio} of $\X$ is
\begin{equation*}
    R_D(\S) := \mathop{{\rm lim\ inf}}\limits_{n \to \infty} \frac{|\S \cap \{ \V_\G \times [-n, n]\} |}{2n N },
\end{equation*}
where $N = |\V_\G|$.

For DTVGS, our goal is also to develop a sampling theorem and give a feasible sampling and reconstruction procedure. By establishing the relations between DTVGS and CTVGS, we propose the lower bounds on the sampling ratio $R_D(\S_\Theta), \S_\Theta \subseteq \S$ (\cref{thm:subset_d}) and the sampling scheme (\cref{subsec:multi_d}) of DTVGS. The significance of giving the sampling theory for DTVGS is that 
\begin{enumerate}[(i)]
    \item Theoretically, we provide the necessary conditions for reconstructing DTVGS from samples. Unlike the sampling of CTVGS, DTVGS are resampled by a rational factor in the time domain.

    \item In practice, once real signals are recorded by sensors, we only have their discretized approximation. The sampling theorem and scheme of DTVGS can help us to compress the signal to the smallest sample size without loss.
\end{enumerate}

\subsection{Sampling for FTVGS}
For an FTVGS $\X$, suppose $\S_v \subseteq \V_\T $ is the sampling set of $\X(v,\cdot)$. The sampling operation on the vectorized form of $\X(v,\cdot)$ can be expressed as $ \X(v,\S_v) = \X(v,\cdot) \mathbf{\Psi}_v^H \in \C^{1\times|\S_v|}, $
where $\mathbf{\Psi}_v \in \{ 0, 1 \}^{ |\S_v| \times T}$ are the sampling matrix corresponding to $\S_v$. 

The sampling set of $\X$ is denoted as $\S=\{ (v, n): v \in \V_\G, n \in \V_\T \}$. We define the projected sampling set of $\S$ in vertex domain as $\S_\G=\{ v: (v, n) \in \S \}$, and the projected sampling set of $\S$ in time domain as $\S_\T = \{ n: (v, n) \in \S \}$. 

The sampling of $\X$ in vertex and time domains respectively are expressed as $\X(\S_\G,\cdot) = \mathbf{\Psi}_\G \X$ and $\X(\cdot,\S_\T) = \X \mathbf{\Psi}_\T^H,$
where $\mathbf{\Psi}_\G \in \{ 0, 1 \}^{|\S_\G| \times N}$ and $\mathbf{\Psi}_\T \in \{ 0, 1 \}^{|\S_\T| \times T}$ are the sampling matrices corresponding to $\S_\G$ and $\S_\T$.

For each $v$, the \emph{sampling ratio} is defined as
\begin{equation*}
    R_F(\S_v) := \frac{|\S_v |}{T} ,
\end{equation*}
and the \emph{total sampling ratio} of $\X$ is given by
\begin{equation*}
    R_F(\S) := \frac{|\S |}{NT}.
\end{equation*}

For FTVGS, a special case of DTVGS, we formulate the sampling theorem in matrix form and provide a proof. The sampling theorem for FTVGS is meaningful in that it completes the framework of the sampling theory of TVGS, and the matrix form makes the theorem simpler and easier to understand compared to the sampling theorems for CTVGS and DTVGS. Also, the sampling and reconstruction procedure applicable to DTVGS applies to FTVGS. Therefore we do not dwell on the sampling scheme in \cref{sec:cri_samp_f} but give a simple example in \cref{subsec:exp1} to help understand our ideas about the sampling scheme.

In addition, we briefly summarise the similarities and differences in the sampling theories and methods for the three types of TVGS in \cref{tab:sum}.

\section{Critical Sampling and Reconstruction of CTVGS} 
\label{sec:cri_samp_c}

\subsection{Critical sampling of CTVGS}
\label{subsec:cri_c}

A key prerequisite of classical sampling theory is that the signal to be sampled is bandlimited, meaning its spectral support consists of finite measures. 

For a CTVGS $\X$ with JFT $\F_\J(\X)$, the $i$-th element $\F_\J(\X)(i, \cdot)$ is supported on a measurable set $\mathcal{F}_i$. We assume that $\mathcal{F}_i$ is a finite union of intervals with known locations. Let $\I_f=\{ i: f \in \mathcal{F}_i\}$ be the index set of nonzero elements of $\F_\J(\X)$ at frequency $f$. Taking $\F_\J(\X)$ in \cref{fig:TVGS} (d) as an example, 
\begin{equation*}
    \I_f=\{ i: f \in \mathcal{F}_i\} \left \{ 
\begin{array}{ll}
\{ 2 \}, & f \in [1, 2) \\ 
\{ 1, 2, 3 \}, & f \in [2, 5) \\
\{ 3 \}, & f \in [5, 6] \\
\emptyset, & {\text{otherwise}.} \end{array}\right.
\end{equation*}

\begin{definition}
\label{df:probth_c}
    Let $\mathcal{F} = \cup_i \mathcal{F}_i$, $B_\T = \mu(\mathcal{F})$ is defined as the projection bandwidth in the time domain, where $\mu(\cdot)$ is the Lebesgue measure. Let $\I = \cup_{f \in \mathcal{F}} \I_f$, $B_\G = |\I|$ is defined as the projection bandwidth in vertex domain. 
\end{definition}

\begin{definition}
\label{df:jotbth_c}
    For any CTVGS $\X$ with spectral functions $\F_\J(\X)$, the joint bandwidth of $\X$ is $B:= \sum_{i=1}^{B_\G} \mu( \mathcal{F}_i) $. Signal $\X$ is a JBL signal when $B < +\infty$.
\end{definition}

We assume that $\F_\J(\X)(i, \cdot)$ are independent (\cref{df:JBLC}), which means that the $\F_\J(\X)(i, f), i \in \I, f \in \mathcal{F}_i$ can be an arbitrary value, provided the bandwidth condition for each $i$ is satisfied. 

For instance, a CTVGS $\X$ is shown in \cref{fig:TVGS} (b). The \cref{fig:TVGS} (c) and (d) are $\F_{FT}(\X)$ and $\F_\J(\X)$, respectively. Then we have $\mathcal{F}_1=[2,5]$, $\mathcal{F}_2=[1, 5] $, $\mathcal{F}_3=[2,6]$, $\mathcal{F}=\cup_{i=1}^{B_\G}\mathcal{F}_i=[1,6]$, and $\I = \{ 1, 2, 3 \}$. The projection bandwidths of $\X$ are $B_\G= |\I| =3$ and $B_\T=\mu(\mathcal{F})=5$. The joint bandwidth of $\X$ is $B=\sum_{i=1}^3 \mu( \mathcal{F}_i) =11$, so $\X$ is a JBL CTVGS.

Since the frequency of the CTVGS lies in the interval $(-\infty, +\infty)$, $B$ is finite if and only if $B_\T$ is finite. Obviously, the relationship between projection bandwidths and the joint bandwidth is
\begin{equation*}
    B \le B_\G B_\T.
\end{equation*}

\begin{definition}
\label{df:JBLC}
    Let $\BL_C(\{ \mathcal{F}_1, \dots, \mathcal{F}_{B_\G} \}) = \{ \X \in L^2(\V_\G \times \T): \F_\J(\X)(i, f) = 0, \forall f \notin \mathcal{F}_i \}$ be the space of JBL CTVGS, in which $\F_\J(\X)(i, f), f \in \mathcal{F}_i$ can be assigned arbitrary values.
\end{definition}

In general, we can sample the signals on all $N$ vertices by their Nyquist rate \cite{nyquist, shannon}. However, in the vertex domain, the GFT allows us to express $\X$ more efficiently by projecting it onto orthogonal bases, enabling us to sample on a smaller subset of vertices. In the time domain, Landau established a lower bound on the sampling density for sampling of each signal $\X(v,\cdot)$, which is determined by the Lebesgue measure of its spectral support \cite{landau1, landau2}. 
For JBL CTVGS, only a portion of the spectral functions of $\F_\J(\X)$ lies within $\I$, and the spectral support of each $\F_\J(\X)(i,\cdot), i \in \I$ does not exceed $\mathcal{F}$. Therefore, $\F_\J(\X)(i,\cdot)$ in \cref{eq:IJFT_c} can be integrated within $\mathcal{F}$ instead of $\R$. A CTVGS $\X \in \BL_C(\{ \mathcal{F}_1, \dots, \mathcal{F}_{B_\G} \})$ thus admits a low-dimensional representation as
\begin{equation}
\label{eq:sep_samp_c}
    \X(v,t) = \frac{1}{2\pi} \sum_{i \in \I} ( \mathbf{u}_i(v) \int_{\mathcal{F}} \F_\J(\X)(i,\Omega) e^{j\Omega t} d\Omega ), v\in \V_\G,t \in \T. 
\end{equation}

Then we have a rank condition of vertices to be sampled in \cref{lem:rank}. Let $\U_\G(\S_\G', \I)$ be the submatrix formed by retaining rows indexed by $\S_\G'$ and columns by $\I$.

\begin{lemma}
\label{lem:rank}
    For a signal $\X \in \BL_C(\{ \mathcal{F}_1, \dots, \mathcal{F}_{B_\G} \})$, there exist a projection sampling set in vertex domain $\S_\G' \subseteq \V_\G$ with $|\S_\G'| = B_\G$ satisfying
        \begin{equation*}
        {\rm rank} (\U_\G(\S_\G', \I)) = B_\G,
    \end{equation*}
    such that $\X$ can be reconstructed.
\end{lemma}

\begin{proof}
    According to \cref{eq:JFT_c}, we have$ \F_{FT}(\X) = \U_\G \F_\J(\X)$, where each frequency $f$ corresponds to a graph signal, that is, 
    \begin{equation*}
        \F_{FT}(\X)(\cdot, f) = \U_\G \F_\J(\X)(\cdot, f).
    \end{equation*}
    For a graph signal $\F_{FT}(\X)(\cdot, f), \forall f \in \mathcal{F}$, according to Theorem 1 in \cite{theory}, we have
    \begin{equation}
    \label{eq:S_f}
        {\rm rank} (\mathbf{\Psi}_f \U_\G(\cdot, \I_f)) = |\I_f|
    \end{equation}
    for some sampling matrix as $\mathbf{\Psi}_f$, where $\U_\G(\cdot, \I_f)$ is the submatrix formed by retaining columns indexed by $\I_f$ and all rows. 

    Then we consider sampling on $\F_{FT}(\X)$. To ensure \cref{eq:S_f} holds for every $\forall f \in \mathcal{F}$, there must exist a sampling matrix $\mathbf{\Psi}_\G'$ such that
    \begin{equation*}
        {\rm rank} (\mathbf{\Psi}_\G' \U_\G(\cdot, \I)) = {\rm rank} (\U_\G(\S_\G', \I)) = |\I| = B_\G.
    \end{equation*}

    For recovering, $ \mathbf{\Phi}_\G' \mathbf{\Psi}_\G' \F_{FT}(\X) = \F_{FT}(\X)$, where $\mathbf{\Phi}_\G' = \U_\G(\cdot, \I)(\U_M^H \U_M)^{-1} \U_M^H$, where $\U_M = \mathbf{\Psi}_\G' \U_\G(\cdot, \I)$.
\end{proof}

Based on \cref{eq:sep_samp_c} and \cref{lem:rank}, a \emph{separate sampling scheme} is proposed \cite{ji2019hilbert}, whose main idea is as follows. In the vertex domain, $|\S_\G'| = B_\G$ vertices are selected for sampling following \cref{lem:rank}. The corresponding sampling matrix is $\mathbf{\Psi}_\G'$. 
In the time domain, each $\X(v, \cdot), v \in \S_\G'$ is sampled on a discrete set $\S_\T'$, which is a subset of $\T$ obtained based on $B_\T$. Let $ g(B_\T) = \sum_{i \in \Z} \delta \left( t- \frac{i}{B_\T} \right) $, the sampled signal on $\S_\G' \times \S_\T'$ will be 
\begin{equation*}
    \mathbf{\Psi}_\G' \X \  g(B_\T) 
\end{equation*} 
denoted as $\X(\S_\G', \S_\T') $.
Signal $\X$ is a column vector with functions as elements. Pre-multiplying $\mathbf{\Psi}_\G'$ selects the elements corresponding to $\S_\G'$ from $\X$, while post-multiplying by $g(B_\T)$ samples each element at a rate of $B_\T$. 

The separate sampling scheme helps us get a sampling set $\S' = \{ \S_\G' \times \S_\T' \}$ by sampling $\X$ in vertex domain and time domain respectively, with $D(\S') = (B_\G B_\T)/N$. However, the separate sampling scheme may not give a sampling set with the minimum number of samples.

The correlation between $\X(v, \cdot), v \in \V_\G$ remains unchanged after applying FT, as FT is a linear transform. In other words, spectral functions in $\F_{FT}(\X)$ are as related as $\X(v, \cdot), v \in \V_\G$. \emph{Therefore, we can use GFT to map $\F_{FT}(\X)$ into a more compact spectrum $\F_\J(\X)$ whose spectral functions are independent.} When we want to further reduce the sampling rate calculated based on $\F_{FT}(\X)$, it is natural to sample $\X$ based on $\F_\J(\X)$. As $B$ is the bandwidth defined on the JFT spectrum and $ B \le B_\G B_\T $, we are interested in exploring sampling schemes that achieve a sampling density of $B/N$. 

In the problem of CTVGS sampling, $B_\G < N$ implies that there is a correlation between $\X(v, \cdot), v \in \V_\G$. That is, it is not necessary to sample on all vertices. For whole $\F_\J(\X)$, we get an $\S_\G'$ with $|\S_\G'| = B_\G$. It can be seen that perfect recovery $\F_{FT}(\X) = \mathbf{\Phi}_\G' \mathbf{\Psi}_\G' \F_{FT}(\X)$ can be achieved. Therefore, as long as $\mathbf{\Psi}_\G' \F_{FT}(\X) = \F_{FT}(\X)(\S_\G', \cdot)$ can be stably reconstructed from the samples, $\F_{FT}(\X)$ can be stably reconstructed. That is, $\X$ completely determined by $\X(\S_\G', \cdot)$. $\X(\S_\G', \cdot)$ is also in $\BL_C(\{ \mathcal{F}_1, \dots, \mathcal{F}_{B_\G} \})$ and has joint bandwidth $B$.

Now let us consider sampling on $\X(\S_\G', \cdot)$, whose corresponding sampling set is $\S$. We find a connection between TVGS sampling and MIMO sampling, so we generalize Theorem 1 in \cite{MIMO} to CTVGS. 

To better present the sampling theorem, the concept of stable sampling is first introduced. The inner product on $L^2(\R)$ is $\langle x, y \rangle = \int_{\R} x(t)y^*(t) dt$. The inner product on $\ell^2(\Z)$ is $\langle x, y \rangle = \sum_{n \in \Z} x(n)y^*(n) $. The inner product on the vector space of complex-valued finite-dimensional vectors is $\langle x, y \rangle := \sum_{n=0}^{N-1} x(n)y^*(n)$. The sampling operation of TVGS can be expressed as an inner product\cite{foundations}. We regard a TVGS $\X$ as a vector, where each element is a function or sequence. Then the inner product $\langle \X, \X \rangle$ can be calculated from the vector inner product formula. Then the norm is introduced by the inner product: $\| \X \| = \sqrt{\langle \X, \X \rangle}$. 
\begin{definition}
\label{df:stable}
    \cite{stable,MIMO} Let $\Gamma$ be a subset of TVGS, which is either CTVGS, DTVGS, or FTVGS. Then a set $\S = \{ (v, l) \}$ is called a \emph{stable sampling set} with respect to $\Gamma$ if there exist $0 < A$ and $B < +\infty$ such that
    \begin{equation*}
        A \| \X \|^2 \le \sum_v \sum_l |\X(v,l)|^2 \le B \| \X \|^2
    \end{equation*}
    for any $\X \in \Gamma$.
\end{definition}
Stable sampling means that any error in the observed values on $\X$ will cause a controllable error in the reconstruction signals.

\begin{theorem}
\label{thm:subset_c}
    For a signal $\X(\S_\G', \cdot) \in \BL_C(\{ \mathcal{F}_1, \dots, \mathcal{F}_{B_\G} \})$ with joint bandwidth $B$, suppose that $\Theta \subseteq \S_\G'$, and $\S_\Theta = \{ (v, t_{vz}): v \in \Theta \} \subseteq \S $ is a stable sampling set with $D(\S_\Theta)$. Then
    \begin{equation*}
        D(\S_\Theta) \ge \frac{1}{N} \left( B - \int_\mathcal{F} {\rm rank}(\U_\G(\Theta^c, \I_f))df \right),
    \end{equation*}
    where $\Theta^c$ is the complement of $\Theta$ in $\S_\G'$.
\end{theorem}

\begin{proof}
    According to the definition of GFT in \cref{eq:X_Gf_c}, we get inverse graph Fourier transform (IGFT) of $\X(\S_\G', \cdot)$: 
    \begin{equation}
    \label{eq:channel}
        \X(\S_\G', \cdot) = \U_\G(\S_\G', \I) \F_\G(\X)(\I, \cdot).
    \end{equation}
    
    Therefore, 
    \begin{equation}
    \label{eq:channel_f}
        \F_{FT}(\X)(\S_\G', \cdot) = \U_\G(\S_\G', \I) \F_\J(\X)(\I, \cdot) .
    \end{equation}
    
    Let us consider \cref{eq:channel} and \cref{eq:channel_f} in another way. \cref{eq:channel} shows that $\F_\G(\X)(i, \cdot), i \in \I$ are linearly transformed by multiplying the matrix $\U_\G(\S_\G', \I)$ to obtain $\X(\S_\G', \cdot)$. This expression can be regarded as that of a MIMO channel consisting of linear time-invariant filters, where $\F_\G(\X)(\I, \cdot)$ indicates independent inputs of the channel, and $\X(\S_\G', \cdot)$ indicates outputs of the channel. \cref{eq:channel_f} is the frequency domain expression of \cref{eq:channel}.
    
    Now, we want to deduce the lower bound on the sampling densities from $\F_\J(\X)(\I, \cdot)$, which ensures the stable reconstruction of $\X(\S_\G', \cdot)$. Referring to the analysis of the MIMO sampling problem, Eq. (17) in \cite{MIMO} indicates that 
    \begin{equation*}
        \mathop{{\rm lim\ inf}}\limits_{t \to \infty} \frac{|\S \cap \{ \V_\G \times [-t,t]\} |}{2t} = B - \int_\mathcal{F} {\rm rank}(\U_\G(\Theta^c, \I_f))df.
    \end{equation*}
    Thus \cref{thm:subset_c} must be true.
\end{proof}

According to \cref{eq:channel_f}, at frequency $f$, $\F_\J(\X)(\I_f, f) \in \C^{|\I_f| }$ and $\F_{FT}(\X)(\Theta^c, f) \in \C^{|\Theta^c| }$ are vectors, thus 
\begin{equation*}
    \F_{FT}(\X)(\Theta^c, f) = \U_\G(\Theta^c, \I_f) \F_\J(\X)(\I_f, f).
\end{equation*}
Intuitively, ${\rm rank}(\U_\G(\Theta^c, \I_f))$ represents the number of independent components of $\F_\J(\X)$ at frequency $f$ that can be determined solely from $\F_{FT}(\X(\Theta^c,f))$. 
Matrix $\U_\G(\Theta^c, \I_f)$ varies with frequency $f$ and is always a real matrix. Therefore, 
\begin{equation*}
    \mathop{{\rm ess\ inf}}\limits_{f \in \mathcal{F}} \sigma_{\text{min}} ( \U_\G(\Theta^c, \I_f) ) > 0 ,
\end{equation*}
where ${\rm ess\ inf}$ is the essential infimum of a function, and $\sigma_{\text{min}} (\cdot)$ is the smallest nonzero singular value of a matrix.

\cref{thm:subset_c} provides the lower bounds on the sampling densities of $\X(\Theta,\cdot)$ for all subsets $\Theta \subseteq \S_\G'$. This theorem can be illustrated through an example. In a sensor network, sensors equipped with high-speed ADCs are more costly than those with low-speed ADCs. We can reduce the sampling density of $\X(\Theta, \cdot)$ by increasing the sampling densities of neighboring vertices, instead of assuming that $\X(\Theta^c, \cdot)$ is known for all $t \in \R$ as in \cite{MIMO}. But there is still a lower bound on the sampling density of each $\X(\Theta, \cdot), \Theta \subseteq \S_\G'$, which is given by \cref{thm:subset_c}. 
In particular, we have the following corollaries.
\begin{corollary}
\label{cor:DS_c}
    Under the assumption of \cref{thm:subset_c}, when $\Theta = \S_\G'$, we have 
    \begin{equation*}
        D(\S_\Theta) = D(\S) \ge \frac{B}{N}.
    \end{equation*}
\end{corollary}

Since $\S_\G'$ satisfies \cref{lem:rank}, $D(\S)$ is the total sampling density of $\X$. \cref{cor:DS_c} points out that the lower bound on $D(\S)$ must be no less than the joint bandwidth. 

\begin{corollary}
\label{cor:DSv}
    Under the assumption of \cref{thm:subset_c}, when $\Theta = \{ v \}$, the sampling set on vertex $v$ is $\S_v$, and $\Theta^c = \S_\G' \backslash v$. Then we have
    \begin{equation*}
        D(\S_v) \ge B - \int_\mathcal{F} {\rm rank}(\U_\G(\S_\G' \backslash v, \I_f))df.
    \end{equation*}
\end{corollary}

\cref{cor:DSv} provide a lower bound on the sampling density of each $\X(v, \cdot), v \in \S_\G'$. \emph{It does not mean that $\X$ can be stably recovered by sampling all $\X(v, \cdot), v \in \S_\G'$ at their lowest sampling density simultaneously. After all, the total sampling density should be no less than $B/N$}.

Based on \cref{cor:DS_c}, we introduce the concept of critical sampling, which is not necessarily unique.
\begin{definition}
\label{df:cri_c}
    A stable sampling set $\S$ of $\X \in \BL_C(\{ \mathcal{F}_1, \dots, \mathcal{F}_{B_\G} \})$ is a critical sampling set when $ D(\S) = B/N $. 
\end{definition}

Although Venkataramani \emph{et al.} proved the lower bound of MIMO sampling density in \cite{MIMO}, they neither clarify whether the lower bounds on density are achievable nor give a feasible sampling method. Focus on sampling JBL CTVGS with the lowest sampling density, we have the following theorem.

\begin{theorem}
\label{thm:achiv_c}
    For any $\X \in \BL_C(\{ \mathcal{F}_1, \dots, \mathcal{F}_{B_\G} \})$, there must be a sampling set $\S$ such that $ D(\S) = B/N $ and $\S_\G=\{ v: (v, t_{vz}) \in \S \}$ satisfies $|\S_\G| \le B_\G$. Thus the critical sampling is achievable.
\end{theorem} 

\begin{proof}
    We construct a \emph{multi-band sampling} scheme to obtain such a sampling set and prove that the critical sampling of any JBL CTVGS is achievable. See details in \cref{subsec:multi_c}.
\end{proof}

\subsection{Multi-band sampling scheme for CTVGS}
\label{subsec:multi_c}

\subsubsection{Sampling}

In this subsection, we construct a multi-band sampling scheme that samples $\X$ at the lowest total sampling density $B$.

First, we select a suitable set of vertices $\S_\G'$ based on \cref{lem:rank}.

Then, we consider dividing $\F_\J(\X)$ into $K$ subbands. Let $\P^k := \P_{[f_k, f_{k+1}]}$ be a bandpass pre-filter \cite{foundations}, and $\P^k(\F_\J(\X))$ be the portion of $\F_\J(\X)$ in $[f_k, f_{k+1}]$. Each band $[f_k, f_{k+1}]$ should be the largest interval such that $\P^k(\F_\J(\X))(i, \cdot), i \in \I$ is either fully zero or fully nonzero. 

To show how we divide the bands, the $B_\G$ spectral functions of $\F_\J(\X)$ are shown in \cref{fig:stage} (a). The first and third spectral functions are nonzero in $[f_1, f_2]$, while the second function is zero. Thus the first subband is $[f_1, f_2]$, as shown in \cref{fig:stage} (b). In the remaining part (shown in \cref{fig:stage} (c)), the second and third spectral functions are non-zero in $[f_2, f_3]$. Then \cref{fig:stage} (c) is divided into two parts, \cref{fig:stage} (d) ($[f_2, f_3]$) and \cref{fig:stage} (e) ($[f_3, f_4]$).

\begin{figure} [htbp] 
	\centering
	\includegraphics[width=0.8\columnwidth]{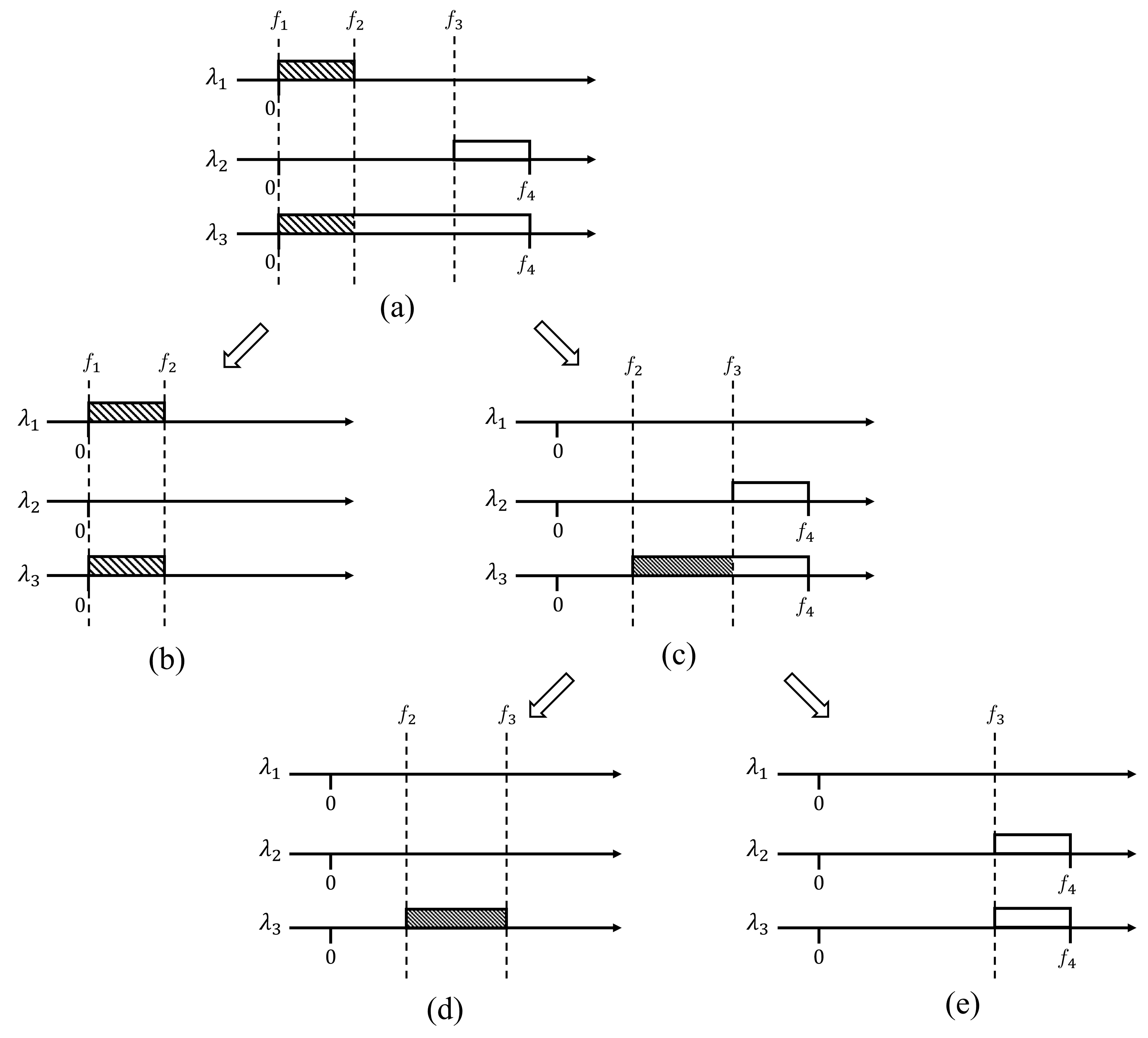}
	\caption{Schematic diagram of dividing subbands. (a) $\F_\J(\X)$ is divided into three parts. (b) Schematic diagram of $\P^1(\F_\J(\X))$ in the interval $[f_1, f_2]$. (c) Schematic diagram of $\F_\J(\X) - \P^1(\F_\J(\X))$, which is divided into two parts. (d) Schematic diagram of $\P^2(\F_\J(\X))$ in the interval $[f_2, f_3]$. (e) Schematic diagram of $\P^3(\F_\J(\X))$ in the interval $[f_3, f_4]$. } 
	\label{fig:stage}
\end{figure}

Within each band, we further discuss how to sample $\P^k(\X)$ based on $\P^k(\F_\J(\X))$.

\emph{Within band $k$}: We first consider sampling on $\P^k(\F_\G(\X))$, where $\P^k(\F_\G(\X)) = \F^{-1}(\P^k(\F_\J(\X)))$. Each $\P^k(\F_\G(\X))(i, \cdot)$ can be sampled at a rate of $(f_{k+1} - f_k)$ Hz, and the corresponding sampling set is $\S_\T^k$.

Now we analyze how to get the sampling set on $\P^k(\X)$ instead of sampling on $\P^k(\F_\G(\X))$. Since the projection operation is linear, the order of the IGFT and projection operation can be exchanged. The following relationship is established
\begin{equation*}
    \U_\G \P^k(\F_\G(\X))= \P^k(\U_\G \F_\G(\X))= \P^k(\X) .
\end{equation*}
Therefore, $\P^k(\X)$ also can be sampled on $\S_\T^k$. Similarly, the sampling operation is also linear, thus
\begin{equation*}
    \U_\G \P^k(\F_\G(\X)) g(f_{k+1} - f_k) = \P^k(\X) g(f_{k+1} - f_k).
\end{equation*}

Then we need to determine which vertices of $\P^k(\X)$ to be sampled. Let $\I_f^k$ be the index set of nonzero elements of $\F_\J(\P^k(\X))$ at frequency $f$, and $\I^k = \cup_f \I_f^k$. 

\begin{corollary}
\label{cor:rank_k}
    Under the assumption of \cref{lem:rank}, within band $k$, there exist an $\S_\G^k \subseteq \S_\G'$ with $|\S_\G^k| = |\I^k|$ satisfying
        \begin{equation*}
        {\rm rank} (\U_\G(\S_\G^k, \I^k)) = |\I^k|,
    \end{equation*}
    such that $\P^k(\X)$ can be reconstructed.
\end{corollary}

\begin{proof}
    We know that $\I^k \subseteq \I$ and matrix $\mathbf{\Psi}_\G' \U_\G(\cdot, \I)$ is full rank. Then there must exist a $\hat{\mathbf{\Psi}}_\G^k$ such that 
\begin{equation*}
    {\rm rank}(\hat{\mathbf{\Psi}}_\G^k \mathbf{\Psi}_\G' \U_\G(\cdot, \I^k) ) = {\rm rank}(\mathbf{\Psi}_\G^k \U_\G(\cdot, \I^k) ) = |\I^k|.
\end{equation*}

    So we can always find such an $\S_\G^k \subseteq \S_\G'$, where $\mathbf{\Psi}_\G^k$ is the corresponding sampling matrix of $\S_\G^k$.

\end{proof}

As a result, we have the multi-band sampled signal at the $k$-th band:
\begin{equation*}
    \mathbf{\Psi}_\G^k \P^k(\X) g(\S_\T^k)
\end{equation*}
denoted as $\tilde{\X}^k(\S_\G^k, \S_\T^k)$. In this band, the sampling rate of each $\P^k(\X)(v, \cdot), v \in \S_\G^k$ is $(f_{k+1} - f_k)$ Hz. The sampling set on $\P^k(\X)$ is recorded as $\S^k = \{ \S_\G^k \times \S_\T^k \}$. 

The process of the multi-band sampling scheme is shown in \cref{fig:sampling_c}. 

\begin{figure} [htbp] 
	\centering
	\includegraphics[width=0.5\columnwidth]{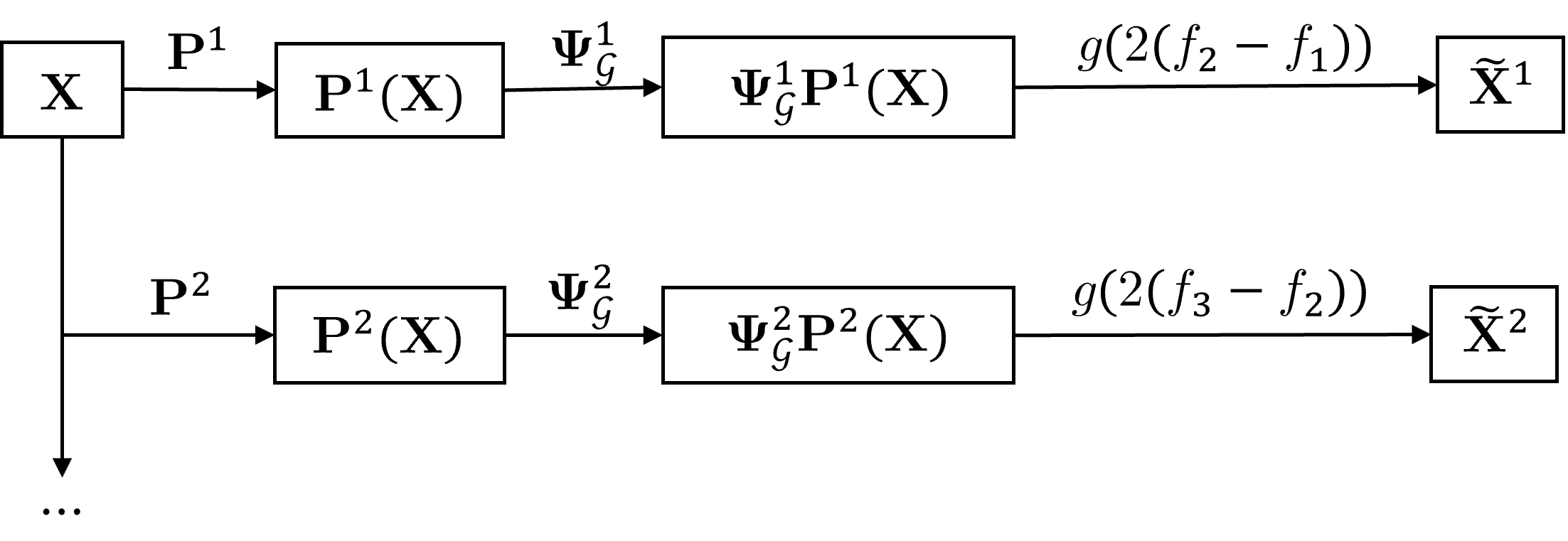}
	\caption{Flow chart of the multi-band sampling scheme for a CTVGS $\X$} 
	\label{fig:sampling_c}
\end{figure}

To more clearly state the idea of the multi-band sampling scheme, we summarize the above process into \cref{alg:multi}.

\begin{algorithm}[htbp] 
\caption{Multi-band sampling scheme}
\label{alg:multi}
\begin{algorithmic}[1]

\REQUIRE $\X$, $\U_\G$

\STATE Calculate $\F_\J(\X)$ according to \cref{eq:JFT_c} and obtain $\I$ according to \cref{df:jotbth_c}.

\STATE Select a suitable set $\S_\G'$, which is the index set of maximal linearly independent rows of $\U_\G(\cdot, \I)$.

\STATE Let $\P^k := \P_{[f_k, f_{k+1}]}$ and divide $\F_\J(\X)$ into $K$ subbands, ensuring that $\P^k(\F_\J(\X))(i, \cdot), i \in \I$ is either fully zero or fully nonzero within each band.

\FOR{$k = 1$ to $K$}
    \STATE  Let $\I_f^k$ be the index set of nonzero elements of $\F_\J(\P^k(\X))$ at frequency $f$, and $\I^k = \cup_f \I_f^k$.
    \STATE Get $\S_\G^k \subseteq \S_\G'$ such that ${\rm rank}(\mathbf{\Psi}_\G^k \U_\G(\cdot, \I^k) ) = |\I^k|$.
    \STATE Get $\tilde{\X}^k(\S_\G^k, \S_\T^k)$ by sampling $\P^k(\X)(v, \cdot), v \in \S_\G^k$ at $(f_{k+1} - f_k)$ Hz (downsampling at $(f_{k+1} - f_k)/f_s$ for DTVGS and FTVGS).
\ENDFOR

\ENSURE Samples $\{ \tilde{\X}^1(\S_\G^1, \S_\T^1), \dots, \tilde{\X}^K(\S_\G^K, \S_\T^K) \}$

\end{algorithmic}
\end{algorithm}

\emph{Critical sampling set}: The sampling set $\S$ for the entire signal $\X$ is the set of the sampling sets of all the bands. Thus the total sampling density of $\X$ is
\begin{equation*}
    D(\S) = \frac{1}{N} \left( \sum_k |\I^k| (f_{k+1}-f_k) \right) = \frac{B}{N}.
\end{equation*}

In addition, we have $\S_\G^k \subseteq \S_\G'$ for all $k = 1, \dots, K$. Then $\S_\G = \cup_k \S_\G^k$, and $|\S_\G| \le |\S_\G'| = B_\G$. Therefore, \cref{thm:achiv_c} holds.

\subsubsection{Recovery}

Each $\tilde{\X}^k(\S_\G^k, \S_\T^k), k = 1, \dots, K$ can be interpolated by $\text{sinc}((f_{k+1}-f_k)t)e^{-j\pi f_kt}$ to get $\mathbf{\Psi}_\G^k \P^k(\X)$. Then $\mathbf{\Psi}_\G^k \P^k(\X)$ is pre-multiplied by $\mathbf{\Phi}_\G^k = \U_\G^k(\cdot, \I^k)((\U_M^k)^H \U_M^k)^{-1}(\U_M^k)^H$, where $\U_M^k = \mathbf{\Psi}_\G^k \U_\G^k(\cdot, \I^k)$, to recover the signal in band $[f_k, f_{k+1}]$, \textit{i.e.}, $\hat{\X}^k = \P^k(\X)$.

We add all the recovered sub-band signals together to get the recovered original signal  
\begin{equation*}
    \hat{\X}= \sum_k \hat{\X}^k =\X.
\end{equation*}
The flow chart of CTVGS recovery is shown in \cref{fig:recover_c}. From the perspective of the spectrum, each recovered signal recovers a part of the spectrum of $\X$. 

\begin{figure} [htbp] 
	\centering
	\includegraphics[width=0.5\columnwidth]{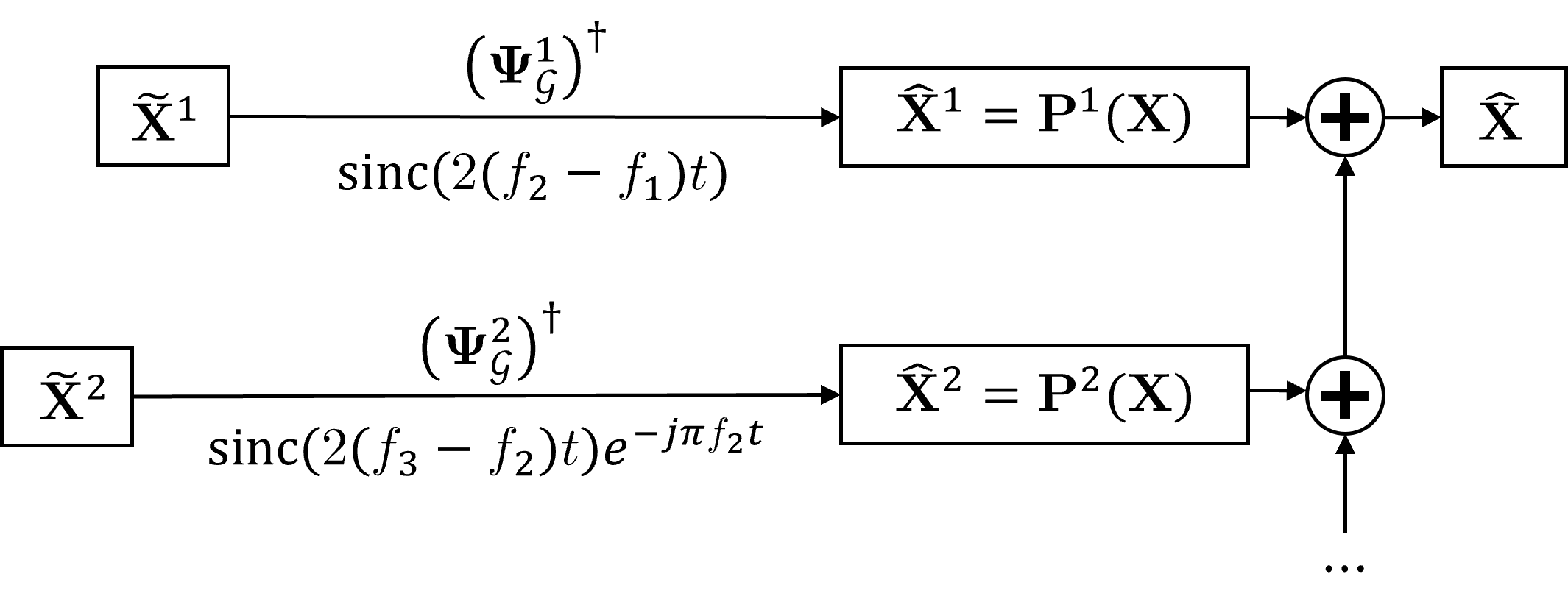}
	\caption{Recovery flow chart of a CTVGS $\X$} 
	\label{fig:recover_c}
\end{figure}

\section{Critical Sampling and Reconstruction of DTVGS} 
\label{sec:cri_samp_d}

\subsection{Critical sampling of DTVGS}
\label{subsec:cri_d}

Assuming that $f_s$ is much higher than the Nyquist rate of the signal on each vertex, the spectra of DTVGS do not overlap. We reduce the redundancy of the original DTVGS by downsampling, which is somewhat different from sampling on a CTVGS.

Downsampling is performed based on the spectral analysis of DTVGS. In \cref{fig:TVGS} (g) and (h), it is observed that both $\F_{FT}(\X)(v, \cdot)$ and $\F_\J(\X)(i, \cdot)$ are continuous periodic functions. Therefore, our analysis focuses on a single period of these functions. Similar to CTVGS, there are projection bandwidths $B_\G, B_\T$, and joint bandwidth $B$ of DTVGS. The relationship between the projection bandwidths and the joint bandwidth is $B \le B_\G B_\T$.

\begin{definition}
\label{df:JBLD}
    Let $\BL_D(\{ \mathcal{F}_1, \dots, \mathcal{F}_{B_\G} \}) = \{ \X \in \ell^2(\V_\G \times \T): \F_\J(\X)(i, f) = 0, \forall f \notin \mathcal{F}_i \}$ be the space of JBL DTVGS, in which $\F_\J(\X)(i, f), f \in \mathcal{F}_i$ can be assigned arbitrary values.
\end{definition}

Then $\X \in \BL_D(\{ \mathcal{F}_1, \dots, \mathcal{F}_{B_\G} \})$ admits a low-dimensional representation as
\begin{equation}
    \X(v,n) = \frac{1}{2\pi} \sum_{i \in \I} ( \mathbf{u}_i(v) \int_{\mathcal{F}} \F_\J(\X)(i, e^{j \omega}) e^{j\omega n} d\omega ), v \in \V_\G, n \in \T .
\end{equation}

The discretization of the signal in the time domain does not affect the correlation between vertices, so \cref{lem:rank} remains valid for DTVGS.

Then we can sample $\X \in \BL_D(\{ \mathcal{F}_1, \dots, \mathcal{F}_{B_\G} \})$ separately. In the vertex domain, we obtain the sampling set $\S_\G'$ according to \cref{lem:rank}. In the time domain, each $\X(v, \cdot), v \in \S_\G'$ can be downsampled with a ratio of at least $B_\T$. Downsampling a sequence with ratio $B_\T = \alpha / \beta, \alpha<\beta (\alpha, \beta \in \mathbb{N})$ is equivalent to upsampling the sequence by $\alpha$ followed by downsampling it by $\beta$, denoted as $(\downarrow \frac{\alpha}{\beta} )$ \cite{r_conv}. Thus, the sampled DTVGS will be
\begin{equation*}
    \mathbf{\Psi}_\G' \X (\downarrow \frac{\alpha}{\beta} )
\end{equation*}
denoted as $\X(\S_\G', \S_\T') $.
Apply the separate sampling scheme\cite{ji2019hilbert} to $\X$ with a total sampling ratio of $R_D(\S_\G' \times \S_\T') = (B_\G B_\T)/N $. 

For JBL DTVGS, the relationship $B \le B_\G B_\T$ still holds, indicating that we aim to sample $\X$ with a total sampling ratio of $B/N$. Once again, the vertices to be sampled need to be constrained as well. Upon satisfying \cref{lem:rank}, we sample $\X(\S_\G', \cdot)$ instead of $\X$. Then we prove a theorem that provides lower bounds on the sampling ratios for all $\Theta \subseteq \S_\G'$, which is not given in \cite{MIMO}.

\begin{theorem}
\label{thm:subset_d}
    For a signal $\X(\S_\G', \cdot) \in \BL_D(\{ \mathcal{F}_1, \dots, \mathcal{F}_{B_\G} \})$ with joint bandwidth $B$, suppose that $\Theta \subseteq \S_\G'$, and $\S_\Theta = \{ (v, n_{vz}): v \in \Theta \} \subseteq \S $ is a stable sampling set with $R_D(\S_\Theta)$. Then
    \begin{equation*}
        R_D(\S_\Theta) \ge \frac{1}{|\Theta|} \left(B - \int_\mathcal{F} {\rm rank}(\U_\G(\Theta^c, \I_f))df \right),
    \end{equation*}
    where $\Theta^c$ is the complement of $\Theta$ in $\S_\G'$. 
\end{theorem}

\begin{proof}
    We can construct a CTVGS $\X_o$ through $\text{sinc}$ function interpolation on $\X(v,n)=\X(v,n T_s), v \in \S_\G'$:
    \begin{equation*}
        \X_o(v,t) = \sum^{+\infty}_{n=-\infty} \X(v,nT_s) \text{sinc}((t-nT_s)\frac{\pi}{T_s}), 
    \end{equation*}
    which is naturally a bandlimited signal. The FT spectrum of $\X_o$ in terms of frequency $f$ is $\F_{FT}(\X_o)(v,f)$.
    
    According to the Poisson summation formula, $\F_{DT}(\X)(\S_\G', \cdot)$ is a periodic replica of $\F_{FT}(\X_o)(v,f)$:
    \begin{equation*}
        \F_{DT}(\X)(v,f) = \sum^{+\infty}_{n=-\infty} \X(v,n) e^{-j2\pi fT_sn} = \frac{1}{T_s} \sum^{+\infty}_{c=-\infty} \F_{FT}(\X_o)(v,f-\frac{c}{T_s}).
    \end{equation*}
    In a single period, we obtain
    \begin{equation}
    \label{eq:surjc}
        \F_{DT}(\X)(v,f) \! = \! \frac{1}{T_s} \F_{FT}(\X_o)(v,f), f \! \in \! [-\frac{f_s}{2}, \frac{f_s}{2}], v \in \S_\G',
    \end{equation}
    which is a surjection. 
    
    Assume that $\X(\S_\G', \cdot)$ can be stably reconstructed with $R_D(\S_\Theta') < \frac{1}{|\Theta|} \left(B - \int_\mathcal{F} {\rm rank}(\U_\G(\Theta^c, \I_f))df \right).$
    From \cref{eq:surjc}, we conclude that there is a method that ensures $\X_o$ is stably reconstructed with \\ $D(\S_\Theta') < \frac{1}{N} \left( B - \int_\mathcal{F} {\rm rank}(\U_\G(\Theta^c, \I_f))df \right),$
    which is impossible according to \cref{thm:subset_c}. The assumption does not hold.
\end{proof}

Since $\U_\G(\Theta^c, \I_f)$ is a real matrix, 
\begin{equation*}
    \mathop{{\rm ess\ inf}}\limits_{f \in \mathcal{F}} \sigma_{\text{min}} ( \U_\G(\Theta^c, \I_f) ) > 0
\end{equation*}
holds for all $f$. Particularly, we have the following corollaries.

\begin{corollary}
\label{cor:RDS}
    Under the assumption of \cref{thm:subset_d}, when $\Theta = \S_\G'$, we have
    \begin{equation*}
        R_D(\S_\Theta) = R_D(\S) \ge \frac{B}{N}.
    \end{equation*}
\end{corollary}

In other words, when the vertices in $\S_\G'^c$ no longer provide new information for stable reconfiguration of $\X$, the total sampling ratio must be no less than $B/N$. If we want to reduce the sampling ratio of $\X(v, \cdot), v \in \S_\G'$, the lower bound on the sampling ratio of each sampled vertex is provided in \cref{cor:RDSv}. 
\begin{corollary}
\label{cor:RDSv}
    Under the assumption of \cref{thm:subset_d}, when $\Theta = \{ v \}$, the sampling set on vertex $v$ is $\S_v$, and $\Theta^c = \S_\G' \backslash v$. Then we have
    \begin{equation*}
        R_D(\S_v) \ge B - \int_\mathcal{F} {\rm rank}(\U_\G(\S_\G' \backslash v, \I_f))df.
    \end{equation*}
\end{corollary}

Note that it is generally not possible to sample all $\X(v, \cdot), v \in \S_\G'$ at their respective lowest rates simultaneously.

In practice, it is more feasible to handle DTVGS than CTVGS. Each component in $\F_\J(\X)$ for DTVGS is a periodic replica of $\F_\J(\X_o)$, so $\F_\J(\X_o)$ and $\F_\J(\X)$ in a single period are both continuous functions supported on a measurable set $\mathcal{F}$. Thus we derive \cref{thm:subset_d} and its corollaries, which provide lower bounds on the sampling ratios of DTVGS. These lower bounds can guide the sampling of CTVGS by multiplying them by $f_s$ in practical experiments.

In addition, the critical sampling of JBL DTVGS can be defined as follows.
\begin{definition}
\label{df:cri_d}
    A stable sampling set $\S$ of $\X \in \BL_D(\{ \mathcal{F}_1, \dots, \mathcal{F}_{B_\G} \})$ is a critical sampling set when $ R_D(\S) = B/N $. 
\end{definition}

Concentrating on whether the ratio satisfying \cref{cor:RDS} is achievable and how to achieve it, we give the following theorem.

\begin{theorem}
\label{thm:achiv_d}
    For any $\X \in \BL_D(\{ \mathcal{F}_1, \dots, \mathcal{F}_{B_\G} \})$, there must be a sampling set $\S$ such that $ R_D(\S) = B/N $ and $\S_\G=\{ v: (v, t_{vz}) \in \S \}$ satisfies $|\S_\G| \le B_\G$. Thus the critical sampling is achievable.
\end{theorem}

\begin{proof}
    Our constructed \emph{multi-band sampling} scheme can prove the critical sampling of any JBL DTVGS is achievable and gives a way to obtain such a sampling set. See details in \cref{subsec:multi_d}.
\end{proof}

\subsection{Multi-band sampling scheme for DTVGS}
\label{subsec:multi_d}

\subsubsection{Sampling}

The main difference between CTVGS and DTVGS is their temporal topology. For the discrete temporal topology, a JBL DTVGS is downsampled in \cref{alg:multi}. The flow of the multi-band sampling scheme is shown in \cref{fig:sampling_d}.

We first select a set of vertices $\S_\G'$ according to \cref{lem:rank} such that $|\S_\G'| = B_\G$. Then, define a projection operator $\P^k := \P_{[f_k, f_{k+1}]}, k = 1, \dots, K$ and divide $\F_\J(\X)$ into $K$ subbands. The band is the largest interval such that $\P^k(\F_\J(\X))(i, \cdot), i \in \I$ is either entirely zero or entirely non-zero within the band.

\emph{Within band $k$}: Let $\I_f^k$ be the index set of nonzero elements of $\F_\J(\P^k(\X))$ at frequency $f$, and $\I^k = \cup_f \I_f^k$. Satisfying \cref{cor:rank_k}, \emph{i.e.}, ${\rm rank}(\mathbf{\Psi}_\G^k \U_\G^k(\cdot, \I^k)) = |\I^k|$, we get the sampling set $\S_\G^k \subseteq \S_\G'$ with $|\S_\G^k| = |\I^k|$.

Then $\P^k(\X)(v, \cdot), v \in \S_\G^k$ can be downsampled at a rate of $\alpha^k / \beta^k =(f_{k+1}-f_k)/f_s$: 
\begin{equation*}
    \mathbf{\Psi}_\G^k \P^k(\X) (\downarrow \frac{\alpha^k}{\beta^k})
\end{equation*}
denoted as $\tilde{\X}^k(\S_\G^k, \S_\T^k)$. The sampling set on $\P^k(\X)$ in band $k$ is recorded as $\S^k = \{ \S_\G^k \times \S_\T^k \}$. 

\emph{Critical sampling set}: Finally, the sampling set $\S$ for the entire signal $\X$ is the set of the sampling sets of all the bands. The total sampling ratio of $\X$ is
\begin{equation*}
    R_D(\S) = \frac{1}{N} \sum_k |\I^k| (f_{k+1}-f_k) = \frac{B}{N}.
\end{equation*}

In vertex domain, we have $\S_\G = \cup_k \S_\G^k$, where $\S_\G^k \subseteq \S_\G'$ for all $k = 1, \dots, K$. Then $|\S_\G| \le |\S_\G'| = B_\G$. Therefore, \cref{thm:achiv_d} holds.

\begin{figure} [htbp] 
	\centering
	\includegraphics[width=0.5\columnwidth]{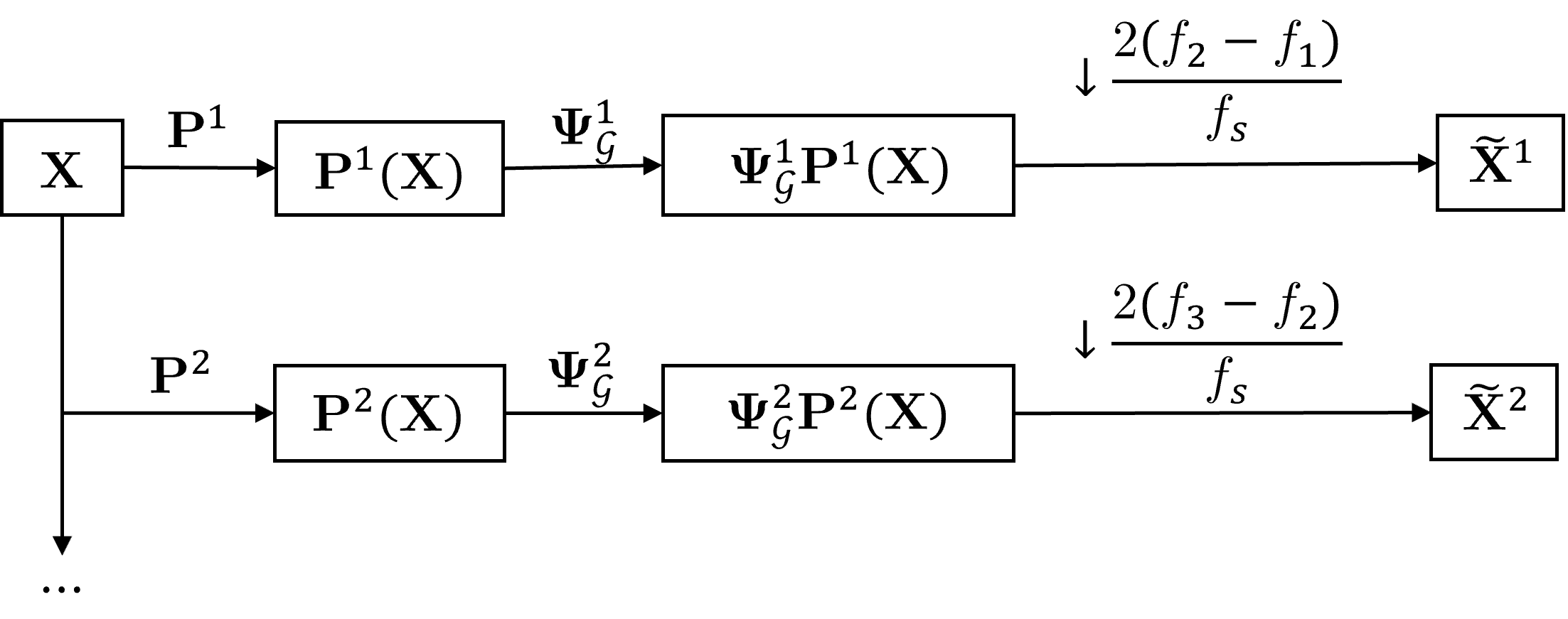}
	\caption{Flow chart of the multi-band sampling scheme for a DTVGS $\X$} 
	\label{fig:sampling_d}
\end{figure}

\subsubsection{Recovery}

The flow chart of the DTVGS recovery is shown in \cref{fig:recover_d}. Each $\tilde{\X}^k$ can be upsampled by $\beta^k / \alpha^k$ (\emph{i.e.}, $\uparrow \frac{\beta^k}{\alpha^k}$) to get $\mathbf{\Psi}_\G^k \P^k(\X)$. Then $\mathbf{\Psi}_\G^k \P^k(\X)$ is pre-multiplied by $\mathbf{\Phi}_\G^k = \U_\G^k(\cdot, \I^k)((\U_M^k)^H \U_M^k)^{-1}(\U_M^k)^H$, where $\U_M^k = \mathbf{\Psi}_\G^k \U_\G^k(\cdot, \I^k)$, to recover the projected sequence $\hat{\X}^k = \P^k(\X)$. In this way, we will get several recovered sequences $\hat{\X}^k, k = 1, \dots, K$, and add them together to get the recovered DTVGS: 
\begin{equation*}
    \hat{\X}= \sum_k \hat{\X}^k =\X.
\end{equation*}

\begin{figure} [htbp] 
	\centering
	\includegraphics[width=0.44\columnwidth]{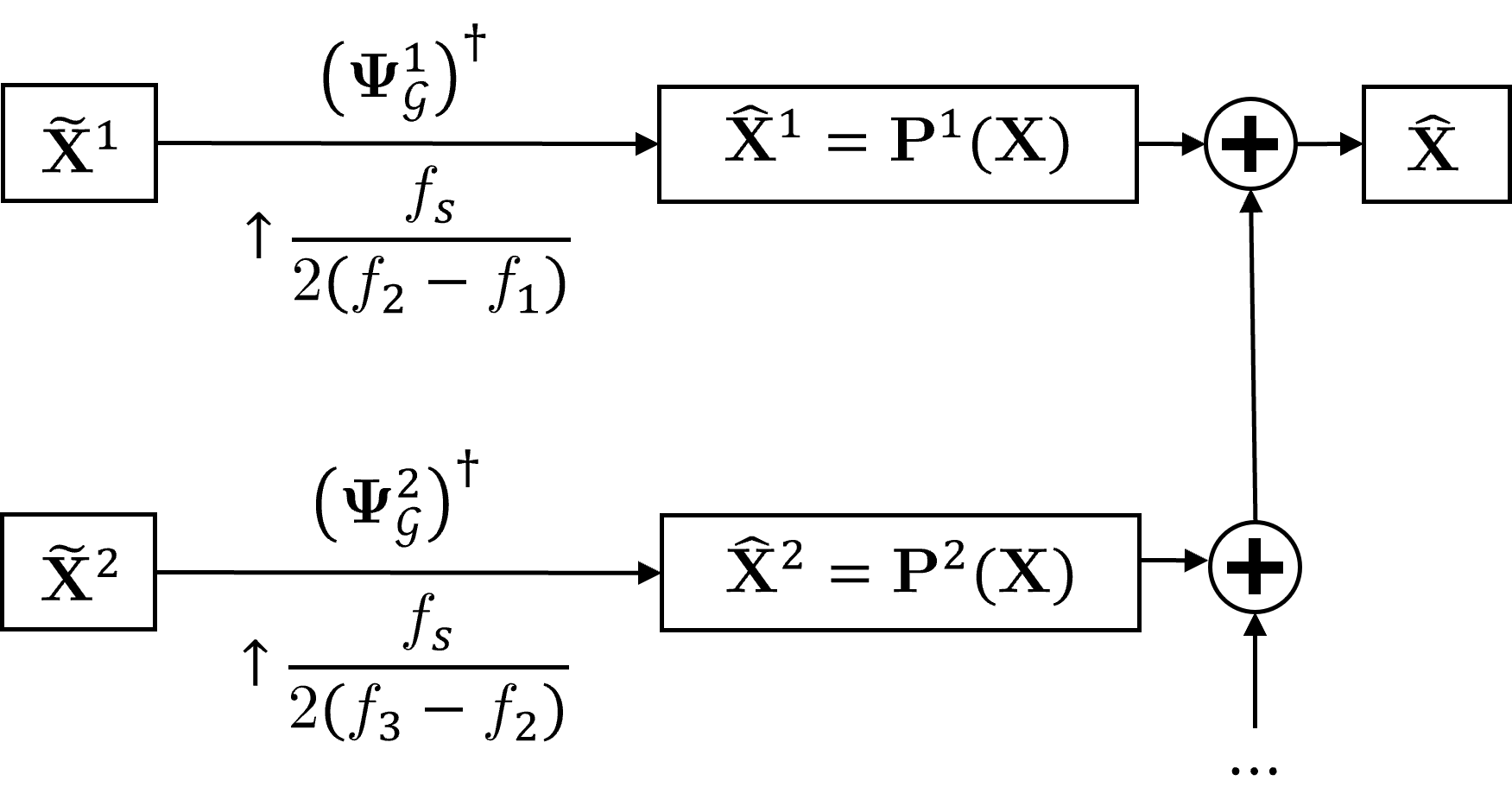}
	\caption{Recovery flow chart of a DTVGS $\X$} 
	\label{fig:recover_d}
\end{figure}

\section{Critical Sampling and Reconstruction of FTVGS} 
\label{sec:cri_samp_f}

\subsection{Critical sampling of FTVGS}
\label{subsec:cri_f}

An FTVGS is modeled as a product graph in \cref{subsec:mod_f}, and its sampling also requires spectral analysis. Therefore, we introduce definitions of bandwidths for FTVGS, which are derived from the concepts of bandwidths in CTVGS and DTVGS. 

The $i$-th row $\F_\J(\X)(i, \cdot)$ is supported on a set $\mathcal{F}_i$. We assume that $\mathcal{F}_i$ is a finite union of frequencies whose locations are known. The index set of nonzero elements of $\F_\J(\X)(\cdot, f)$ is $\I_f$.

\begin{definition}
\label{df:probth_f}
    \emph{\cite{Yu}} Let $\mathcal{F} = \cup_i \mathcal{F}_i$, $B_\T = |\mathcal{F}|$ is defined as the projection bandwidth in time domain. Let $\I = \cup_f \I_f$, $B_\G = |\I|$ is defined as the projection bandwidth in the vertex domain. 
\end{definition}

\begin{definition}
\label{df:jotbth_f}
    An FTVGS $\X$ is a JBL signal when $\F_\J(\X)$ has $B = ||\F_\J(\X) ||_0 < NT$ nonzero elements, where $B$ is the joint bandwidth.
\end{definition}

The relationship between the projection bandwidths and joint bandwidth can be easily obtained: 
\begin{equation*}
       {\rm max}(B_\T, B_\G) \le B \le B_\G B_\T.
\end{equation*}
For example, when $\F_\J(\X)$ is a diagonal matrix with all nonzero diagonal entries, $B<NT$, but $B_\G=N, B_\T=T$.

\begin{definition}
\label{df:JBLF}
    Let $\BL_F(\{ \mathcal{F}_1, \dots, \mathcal{F}_{B_\G} \}) = \{ \X \in \C^{N\times T}: \F_\J(\X)(i, f) = 0, \forall f \notin \mathcal{F}_i \}$ be the space of JBL FTVGS, in which $\F_\J(\X)(i, f), f \in \mathcal{F}_i$ can be assigned arbitrary values.
\end{definition}

A JBL FTVGS gives a low-dimensional representation as
\begin{equation*} 
    \X = \U_\G(\cdot, \I) \F_\J(\X)(\I, \mathcal{F}) (\U_\T(\cdot, \mathcal{F}))^T ,
\end{equation*}
and the vector form
\begin{equation*} 
    \x = (\U_\G(\cdot, \I) \otimes \U_\T(\cdot, \mathcal{F})) \F_\J(\x)(\I \times \mathcal{F}).
\end{equation*}

As the special case of DTVGS, FTVGS also follows \cref{lem:rank}. Moreover, by applying \cref{lem:rank} to FTVGS in the time domain, we obtain $\S_\T' \subseteq \V_\T$ with $|\S_\T'| = B_\T$ satisfying ${\rm rank} (\U_\T(\S_\T', \mathcal{F})) = |\mathcal{F}| = B_\T$.

Then, a separate sampling strategy is used to sample JBL FTVGS \cite{sampling2018}. By performing elimination on $\U_\G(\cdot, \I)$ and $\U_\T(\cdot, \mathcal{F})$ separately, we can obtain the sampling sets $\S_\G'$ and $\S_\T'$, such that $| \S_\G' | = B_\G$ and $| \S_\T' | = B_\T$. The expression of sampling $\X$ is as follows: 
\begin{equation*}
    \X(\S_\G', \S_\T') = \mathbf{\Psi}_\G' \X \mathbf{\Psi}_\T'^H = \mathbf{\Psi}_\G' \U_\G(\cdot, \I) \F_\J(\X)(\I, \mathcal{F}) (\U_\T(\cdot, \mathcal{F}))^T \mathbf{\Psi}_\T'^H,
\end{equation*}
where $\mathbf{\Psi}_\T'$ and $\mathbf{\Psi}_\G'$ are sampling matrices of sampling sets $\S_\T'$ and $\S_\G'$. The vector form of $\X(\S_\G', \S_\T')$ can be expressed as
\begin{equation*} 
    \x(\S_\G' \times \S_\T') = ( \mathbf{\Psi}_\G' \U_\G(\cdot, \I) \otimes \mathbf{\Psi}_\T' \U_\T(\cdot, \mathcal{F}) ) \F_\J(\x)(\I \times \mathcal{F}) ,
\end{equation*}
where $\x(\S_\G' \times \S_\T')$ is the subvector of $\x$ corresponding to rows indexed by $\S_\G' \times \S_\T'$. The total sampling ratio is $ R_F(\S_\G' \times \S_\T') = (B_\G B_\T ) /(NT)$. However, it is important to note that the separate sampling scheme does not guarantee the minimum number of samples in all cases.

\emph{For JBL FTVGS, JFT gets a more compact spectrum. Therefore, we analyze the sampling ratio from the perspective of the joint time-vertex domain.} When considering $\X$ in matrix form, the necessary conditions and proofs for stable sampling will differ significantly from those of CTVGS and DTVGS.

Then the following results give the necessary conditions for stable sampling. We know that when $\S_\G'$ satisfies \cref{lem:rank}, $\X$ completely determined by $\X(\S_\G', \cdot)$, whose joint bandwidth is $B$. So we sample $\X(\S_\G', \cdot)$ with stable sampling set $\S$.

\begin{theorem}
\label{thm:subset_f}
    For a signal $\X(\S_\G', \cdot) \in \BL_F(\{ \mathcal{F}_1, \dots, \mathcal{F}_{B_\G} \})$ with joint bandwidth $B$, $\x(\S_\G' \times \V_\T) = \emph{\text{vec}}(\X(\S_\G', \cdot)) \in \C^{B_\G T}$ is the corresponding vectorized form. Suppose that $\Theta \subseteq \S_\G'$, $\Theta^c$ is the complement of $\Theta$ in $\S_\G'$, and $\S_\Theta = \{ (v, n): v \in \Theta \} \subseteq \S $ is a stable sampling set with $R_F(\S_\Theta)$. Then
    \begin{equation*}
        R_F(\S_\Theta) \ge \frac{1}{|\Theta|T} \left( B - {\rm rank}( \mathbf{\Psi}_\theta^c \U_\J(\S_\G' \times \V_\T, \cdot) \mathbf{\Psi}_j^H ) \right),
    \end{equation*}
    where $\mathbf{\Psi}_\theta^c \in \{0, 1\}^{(B_\G-|\Theta|)T \times B_\G T}$ is the sampling matrix corresponding to $\{ \Theta^c \times \V_\T \}$, $\mathbf{\Psi}_j \in \{0, 1\}^{B \times NT}$ is a sampling matrix that corresponds to nonzero elements of $\F_\J(\x)$.
\end{theorem}

The proof of \cref{thm:subset_f} can be found in \cref{pf_th_f}.

Upon satisfying \cref{lem:rank}, \cref{thm:subset_f} proves the lower bound on the sampling ratios of all $\X(\Theta, \cdot),\Theta \subseteq \S_\G'$. In the vector form, we have $\x(\Theta^c \times \V_\T) = \text{vec}(\X(\Theta^c, \cdot))$ and 
\begin{equation*}
    \x(\Theta^c \times \V_\T) =  \mathbf{\Psi}_\theta^c \U_\J(\S_\G' \times \V_\T, \cdot) \mathbf{\Psi}_j^H \cdot \F_\J(\x)(\mathcal{N}).
\end{equation*}
Therefore, ${\rm rank}( \mathbf{\Psi}_\theta^c \U_\J(\S_\G' \times \V_\T, \cdot) \mathbf{\Psi}_j^H )$ is the number of independent components of $\F_\J(\x)$ that can be determined from the information of $\X(\Theta^c, \cdot)$ alone.

In particular, we have the following corollaries.

\begin{corollary}
\label{cor:RFS}
    Under the assumption of \cref{thm:subset_f}, when $\Theta = \S_\G'$, we have
    \begin{equation*}
        R_F(\S_\Theta) = R_F(\S) = \frac{|\S|}{NT} \ge \frac{B}{NT} .
    \end{equation*}
\end{corollary}

When $\S_\G'$ satisfies \cref{lem:rank}, $\X$ completely determined by $\X(\S_\G', \cdot)$. Set $\S$ is the stable sampling set of $\X(\S_\G', \cdot)$. That is, $\X$ can be stably reconstructed from $\S$. Thus $R_F(\S) = |\S|/(NT)$.

If $\X(v, \cdot), v \in \S_\G'$ are correlated, the sampling ratio of each $\X(v, \cdot)$ can be reduced by increasing the sampling ratios of the related sequences. Such conversion limits are given in \cref{cor:RFSv}. Of course, the total sampling ratio should be no less than $B/(NT)$ to ensure stable reconstruction. 

\begin{corollary}
\label{cor:RFSv}
    Under the assumption of \cref{thm:subset_f}, when $\Theta = \{ v \}$, the sampling set on vertex $v$ is $\S_v$, and $\Theta^c = \S_\G' \backslash v$. Then we have
    \begin{equation*}
        R_F(\S_v) \ge \frac{1}{T} \left( B- {\rm rank}( \mathbf{\Psi}_\theta^c \U_\J(\S_\G' \times \V_\T, \cdot) \mathbf{\Psi}_j^H ) \right).
    \end{equation*}
\end{corollary}

Moreover, \cref{cor:RFS} lead to the definition of a critical sampling set of FTVGS. 

\begin{definition}
\label{df:cri_f}
    A stable sampling set $\S$ of $\X \in \BL_F(\{ \mathcal{F}_1, \dots, \mathcal{F}_{B_\G} \})$ is a critical sampling set when $ R_F(\S) = \frac{B}{NT} $. 
\end{definition}

As described in \cref{subsec:mod_f}, FTVGS is a special case of DTVGS. Thus the multi-band sampling scheme in \cref{subsec:multi_d} also applies to FTVGS. Additionally, we provide an alternative \emph{joint sampling} scheme in \cite{Yu}, which is only applicable to FTVGS. Both the multi-band sampling and joint sampling schemes can prove that critical sampling of any JBL FTVGS is achievable. 

\begin{theorem}
\label{thm:achiv_f}
    For any $\X \in \BL_F(\{ \mathcal{F}_1, \dots, \mathcal{F}_{B_\G} \})$, there must be a sampling set $\S$ such that $ R_D(\S) = B/N $ and $\S_\G=\{ v: (v, t_{vz}) \in \S \}$ satisfies $|\S_\G| \le B_\G$. Thus the critical sampling is achievable.
\end{theorem}

So far, we have introduced the sampling theories and methods for CTVGS, DTVGS, and FTVGS, respectively. We summarise the similarities and differences for the three types of TVGS in \cref{tab:sum}.

\begin{table}
    \centering
    \caption{Similarities and differences in sampling theories and methods for three types of TVGS}
    \begin{tabular}{llll}
        \hline
          & CTVGS & DTVGS & FTVGS \\
         \hline
         \multirow{3}*{Similarities} & \multicolumn{3}{l}{ All three types of TVGS obey the time-vertex graph signal framework \cite{timevertex}, \emph{i.e.}, $\G \times \T$. } \\
         ~ &  \multicolumn{3}{l}{ \makecell[l]{The definitions of $D(\cdot)$, $R_D(\cdot)$, and $R_F(\cdot)$ is fundamentally the same, \emph{i.e.}, the ratio of \\ the number of total samples to the number of total data. }} \\

         \hline
         Temporal topology & $\R$ (infinite, uncountable) & $\Z$ (infinite, countable) & $\G_\T$ (finite, countable) \\
         $\X(i,\cdot)$ & $L^2$ function & $\ell^2$ sequence & finite sequence \\
         JFT & FT + GFT & DTFT + GFT & DFT + GFT \\
         $\F_\J(\X)(i,\cdot)$ & $L^2$ function & periodic $L^2$ function & finite sequence \\
         $|\mathcal{F}|$ & $\infty$ & $\infty$ & finite \\
         \hline
         Sampling methodology & uniform sampling & \makecell[l]{uniform resampling \\ by a rational factor} & \makecell[l]{uniform resampling \\ by a rational factor} \\
         \hline
         \makecell[l]{Practical significance \\ of multi-band sampling} & \makecell[l]{ multi-band sampling for CTVGS \\ aim to guide the process of capturing \\ signals from the real world with \\ minimal samples without loss.} & \multicolumn{2}{l}{ \makecell[l]{multi-band sampling for DTVGS and FTVGS \\ help us to compress the signal to the smallest \\ sample size without loss.} } \\
         \hline
    \end{tabular}
    \label{tab:sum}
\end{table}

On the high level, the three types of TVGS have common features. All of them are processing signals on a product of the graph and temporal domains. When delving into technical specifics, the differences in temporal domains prohibit us from handling certain aspects uniformly. 
\begin{enumerate}[(i)]
    \item Dimensionality: The dimensions of temporal topologies of CTVGS and DTVGS are infinite, so the sampling density and sampling ratio are defined based on limits. Moreover, the time-domain topology $\R$ of CTVGS is uncountable, while the time-domain topology $\Z$ of DTVGS is countable. The time topology of FTVGS is a finite cyclic graph, and the sampling ratio can be defined with a simpler finite ratio.

    \item Sampling methodology: Due to the finite dimensionality of the signal space for FTVGS, sampling, and recovery are performed via usual matrix operations. On the other hand, sampling for CTVGS and DTVGS relies on classical signal processing theory \cite{foundations} on bandlimited functions. For example, we modify the classical Nyquist-Shannon theory to CTVGS. Moreover, sampling for the DTVGS follows the digital signal processing.
    
\end{enumerate}




\section{Examples and Experiments} \label{sec:exp}

The applicability of our sampling schemes is illustrated with different data. The outcomes of these experiments reveal the advantages of joint analysis in TVGS, specifically in reducing the required sampling ratio.

\subsection{Examples of FTVGS}
\label{subsec:exp1}

We show our multi-band sampling scheme on a constructed FTVGS with $N=4$ and $T=4$. The FTVGS in a single period is written in matrix form: 
\begin{equation*}
    \X=\left[\begin{matrix}
    0.4190 &   0.3120 &   0.4382 &   0.5452 \\
    0.3459 &   0.2389 &   0.3651 &   0.4722 \\
    0.3785 &   0.2767 &   0.4194 &   0.5212 \\
    0.2403 &   0.1281 &   0.2378 &   0.3501 
    \end{matrix} \right].
\end{equation*}

Let the temporal topology of FTVGS be a directed cyclic graph, and the topology of $\J$ is shown in \cref{fig:TVGS} (i). The corresponding frequency coefficient of $\X$ is
\begin{equation*}
    \F_\J(\X) \! = \! \left[ \begin{matrix}
    0 &  -0.0384 - 0.4665j &  2.8443 &  -0.0384 + 0.4665j \\
    0 &  0.0306 + 0.0159j &  -0.4522 &   0.0306 - 0.0159j \\
    0 &  0                &   0.3579 &   0 \\
    0 &  0                &   0      &   0
    \end{matrix} \right].
\end{equation*}
    
Then we have $\mathcal{F}=\{ 2, 3, 4 \}$ and $\I=\{ 1, 2, 3 \}$. $\X$ is a JBL signal with $B= 7$, $B_\T=3$, $B_\G=3$. Obviously, JFT gives a smaller overall rate as $B<B_\T B_\G$. Thus we can sample fewer points without losing information.

\emph{Sampling and reconstruction}: For this FTVGS, our sampling process can be divided into two stages. In Stage 1, for $\P^1(\X^1)$ within sub-band 1, the joint spectrum is zero for all columns except the third one. In the graph domain, we have $\I^1 = \{ 1, 2, 3 \}$ and obtain $\S_\G^1 = \{1, 2, 3\}$ satisfying ${\rm rank}(\mathbf{\Psi}_\G^1 \U(\cdot, \I^1)_\G^1)) = |\I^1|$ through Gaussian elimination. In the time domain, we have $\alpha^1 / \beta^1 = 1/4$ and $\S_\T^1 = \{ 1 \}$. Note that in the time domain, any of the four instants can be selected. Then, the sampling set $\S^1 = \{ \S_\G^1 \times \S_\T^1 \}$ for $\P^1(\X^1)$ is shown in \cref{fig:exp1_Sa} (a). 

In Stage 2, for the signal $\X^2 = \P^2(\X^2)$ within sub-band 2, its joint spectrum is zero for all columns except columns $2, 4$. In the graph domain, we obtain $\S_\G^2 = \{1, 3\}$. In the time domain, we have $\alpha^2 / \beta^2 = 1/2$ and $\S_\T^2 = \{ 1, 3 \}$. Note that in the time domain, it works uniform downsampling by a factor of 2. Therefore, sampling set $\S^2 = \{ \S_\G^2 \times \S_\T^2 \}$ for $\P^2(\X^2)$ is shown in \cref{fig:exp1_Sa} (b).

\begin{figure}[htbp]
    \centering    
    \subfigure[]
    {
	    \includegraphics[scale=0.5]{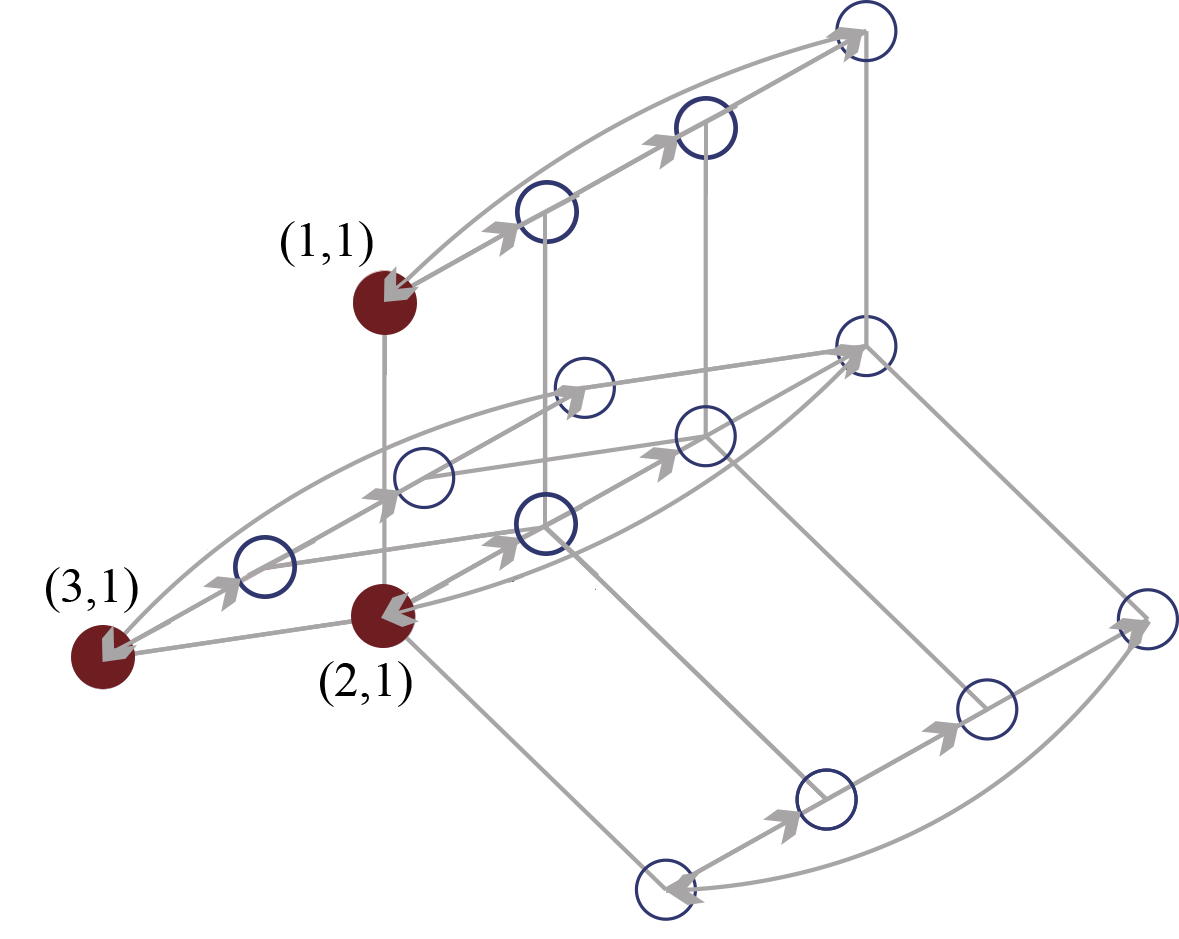}   
	}
    \subfigure[]
    {
	    \includegraphics[scale=0.5]{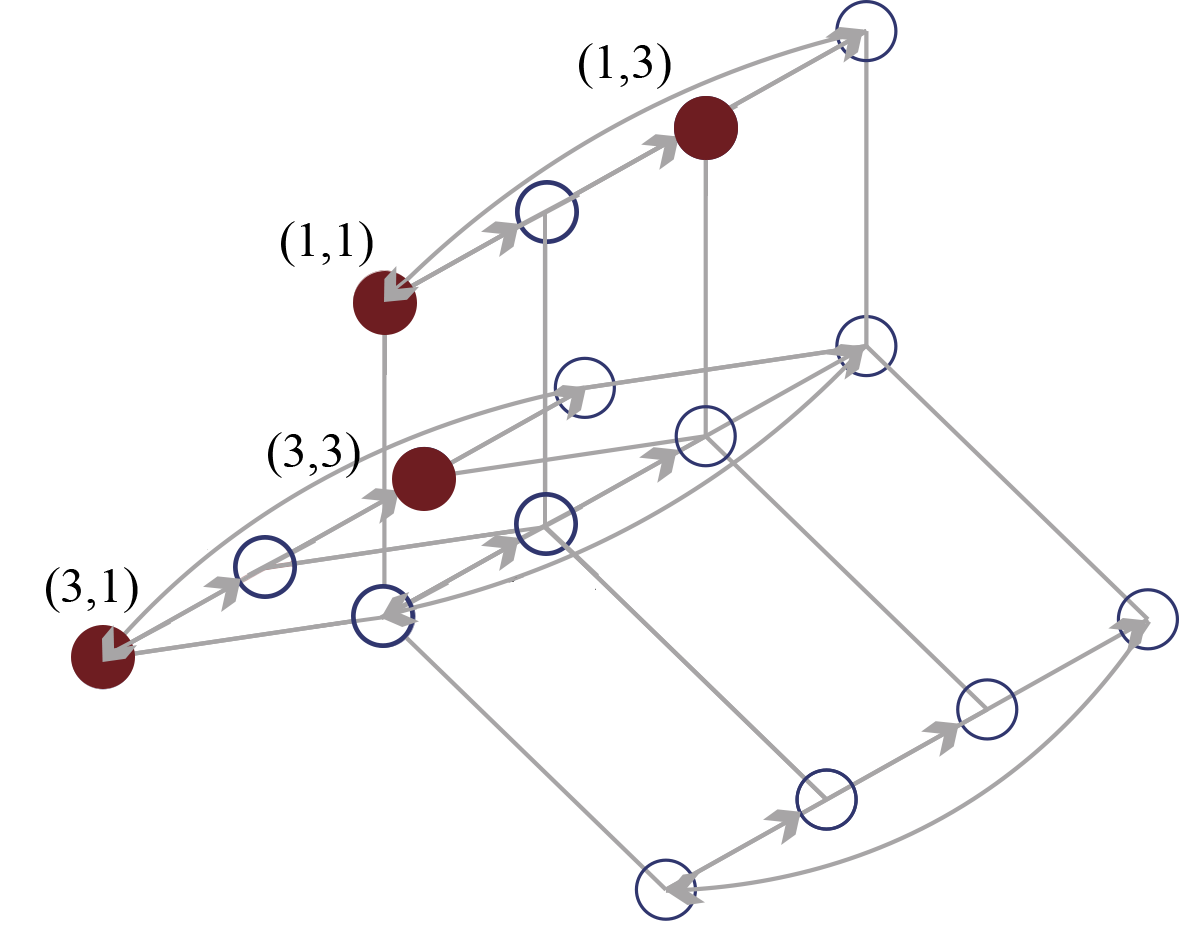}   
    }
    \caption{The schematic diagrams of the critical sampling set $\S$. (a) is the sampling set $\S^1$ of $\P^1(\X^1)$. (b) is the sampling set $\S^2$ of $\P^2(\X^2)$. }
\label{fig:exp1_Sa} 
\end{figure}

Now satisfies $R_F(\S)=7 /16$ and $|\S_\G|=3$, so it is the critical sampling set. The reconstructed signal $\hat{\X} = \X$ is obtained according to the flow of \cref{fig:recover_d}, so we perfectly restored the original signal.

Then we discuss the lowest sampling ratio on each $\X(v,\cdot)$ and the impact when decreasing the sampling ratio on a certain vertex. According to \cref{cor:RFSv}, we consider the lowest sampling ratio of each $\Theta =\{ v \} \in \S_\G$. For $\{ v_1 \}$, $\S_\G \backslash \Theta =\{ v_2, v_3 \}$, $\mathbf{\Psi}_\theta^c$ and $\mathbf{\Psi}_j$ select rows 5-12 and columns 1, 3, 4, 5, 7, 8, 12 of $\U_\J$, respectively. So we have $R_F(\S_1) \ge (7-6)/4  =1/4$. Similarly, $R_F(\S_2) \ge 1/4$ and $R_F(\S_3) \ge 3/4$. We remark that $R_F(\S_1) = 1/4$, $R_F(\S_2) = 1/4$ and $R_F(\S_3) = 3/4$ cannot be achieved at the same time, otherwise $\X$ cannot be perfectly recovered.

As the choices of $\S_\G^1$ and $\S_\G^2$ are not unique in this example, the critical sampling set for $\X$ is also not unique. We can adjust the choices of $\S_\G^1$ and $\S_\G^2$ to increase or decrease the samples at certain vertex. For example, if we change $\S_\G^2$ to $\{2,3\}$, then $\mathbf{\Psi}_\G^2$ also satisfies ${\rm rank}(\mathbf{\Psi}_\G^2 \U(\cdot, \I^2)_\G^2)) = |\I^2|$. The critical sampling set $\S_b$ in this case is shown in \cref{fig:exp1_Sb}. Compared with $\S$, $\S_b$ increases the sampling ratios of signals on $\{ v_2\}$, thereby reducing the sampling ratio of the signal on vertex $\{ v_1 \}$.

\begin{figure}[htbp]
    \centering    
    \subfigure[]
    {
	    \includegraphics[scale=0.5]{Fig//exp1_Sa_1.png}   
	}
    \subfigure[]
    {
	    \includegraphics[scale=0.5]{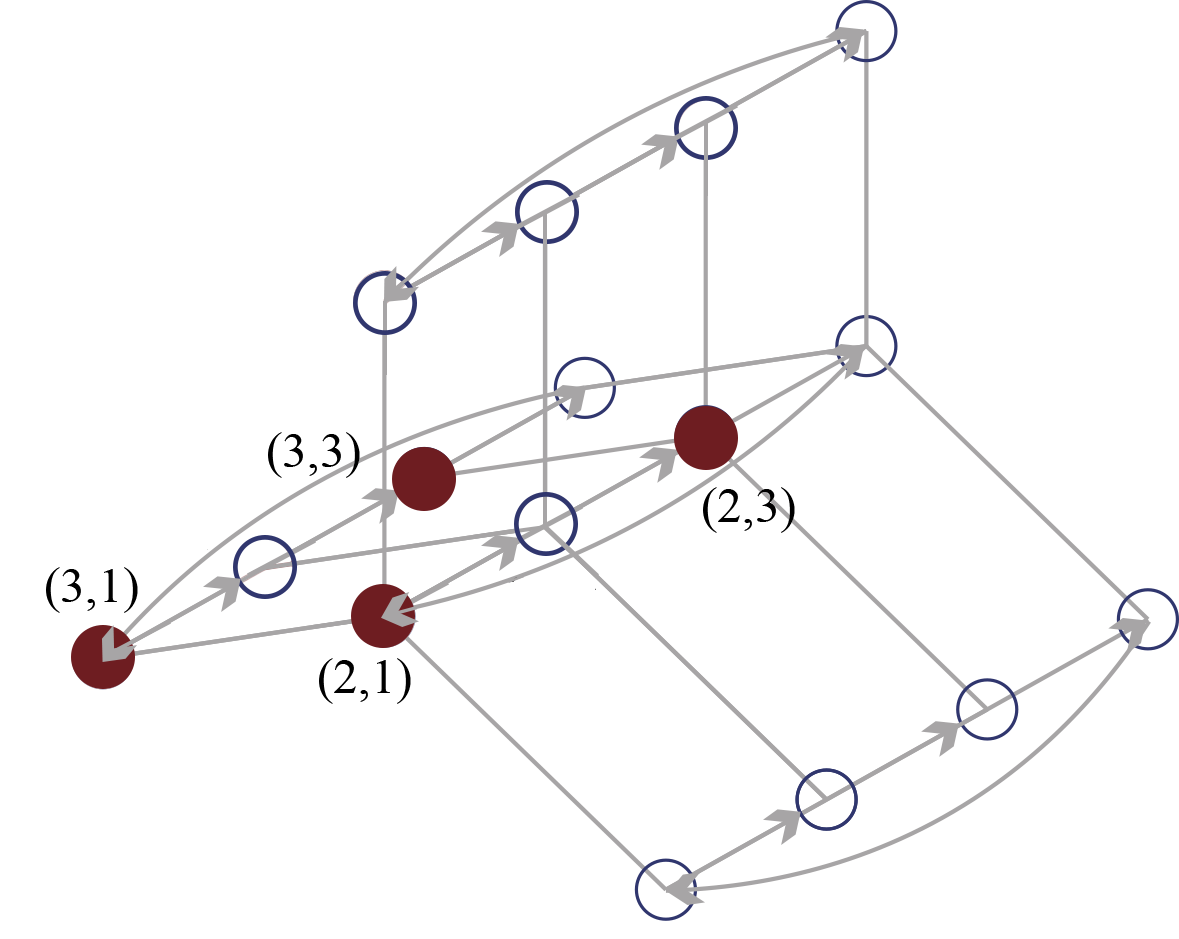}   
    }
    \caption{The schematic diagrams of the critical sampling set $\S_b$. (a) is the sampling set $\S^1$ of $\P^1(\X^1)$. (b) is the sampling set $\S^2$ of $\P^2(\X^2)$. }
\label{fig:exp1_Sb} 
\end{figure}

\subsection{Real datasets on multi-band sampling}
\label{subsec:exp2}

We test our multi-band sampling scheme on two real-world datasets:
\begin{itemize}
    \item EEG: This EEG dataset consists of 32 channels of electroencephalogram signals collected at $f_E = 128$ Hz from a subject in a visual attention task from the EEGLAB tutorial\cite{EEG}. Each channel is regarded as a vertex. 
    
    \item METR-LA: This traffic data METR-LA is collected from loop detectors in the highway of Los Angeles County \cite{data_traffic}. We chose data collected by 207 sensors from March 1st, 2012 to June 27th, 2012 for this experiment. Each sensor is regarded as a vertex with $T_L = 300$ s.
\end{itemize}

\subsubsection{Test on EEG}

    We denote the data of every 1024 consecutive timesteps of the EEG data as $\X$ for simulation. The dataset is divided into 100 DTVGS in total. That is, each $\X$ is a DTVGS with $N=32$ vertices, where each vertex relates to a discrete sequence of 1024 timesteps.
    
    \emph{Graph construction}: A graph typically encodes the similarity between vertices. We construct an adjacency matrix $\W_\G$ based on the correlation coefficients between the signals on vertices and build a graph $\G$ \cite{Structure} as shown in \cref{fig:exp2_W_EEG}. The short duration of $\X$ allows us to assume that the graph structure of the DTVGS is unchanged.
    
    \begin{figure} [htbp] 
	\centering    
        \includegraphics[scale=0.4]{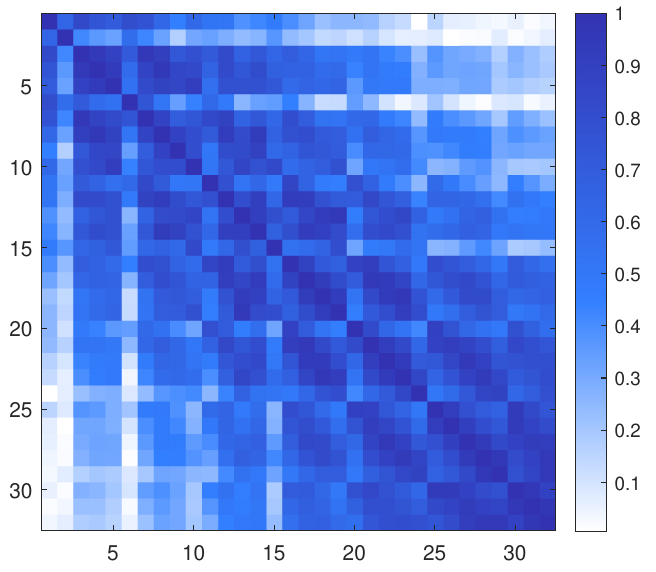}
	\caption{ Weighted adjacency matrix $\W_\G$ of EEG.} 
	\label{fig:exp2_W_EEG}
    \end{figure}

    We show the energy of a DTVGS $\X$ along with the energy of the results obtained by applying FFT, GFT, and JFT to $\X$ in \cref{fig:exp2_signal_EEG}. To obtain a strictly bandlimited signal, we apply a low-pass filter to each spectrum (each row) in $\F_\J(\X)$, retaining no less than $90\%$ of its energy. We keep the $B_\G$ rows with the highest energy while setting the spectral coefficients of the remaining rows to zero. The IJFT is then performed on the resulting coefficients, yielding a JBL DTVGS denoted as $\X_{\rm JBL}$. Taking the signal in \cref{fig:exp2_signal_EEG} as an example, with $B_\G=15$, $\X_{\rm JBL}$ retains $85.88\%$ of the energy in $\X$, with bandwidths of $B=263$ and $B_\T=27.5$.
        
    \begin{figure} [htbp] 
	    \centering
	    \includegraphics[width=0.7\columnwidth]{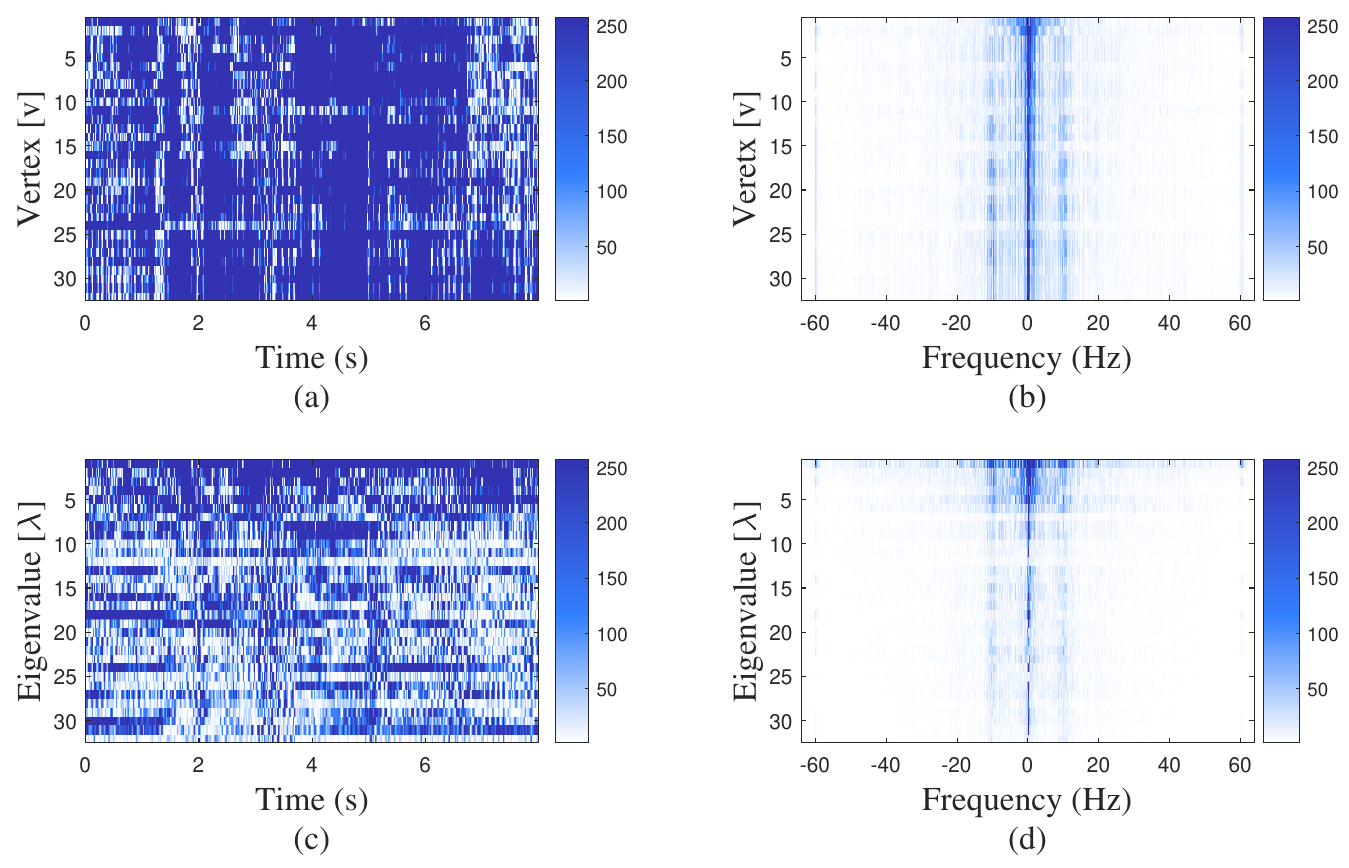}
	    \caption{Energy analysis of $\X$ getting from EEG. (a) is the energy of the original signal $\X$, whose each row is a sequence on a vertex, and each column is a graph signal at a certain instant. (b), (c), and (d) are the energy of $\F_{DT}(\X)$, $\F_\G(\X)$, and $\F_\J(\X)$, respectively.}
	\label{fig:exp2_signal_EEG}
    \end{figure}

    \emph{Sampling and reconstruction}: After compressing DTVGS in the time-vertex domain, the multi-band sampling scheme described in \cref{subsec:cri_d} is tested on $\X_{\rm JBL}$, and the corresponding recovered signal is recorded as $\hat{\X}$. 
    
    For instance, the signal in \cref{fig:exp2_signal_EEG} can be sampled with a total sampling ratio of $R_D(\S) = B/N = 0.0642$, less than $0.1007$ of separate sampling. The set $\S_\G$ is shown in \cref{fig:exp2_G_SG}. We exactly get the critical sampling set.
    
    \begin{figure} [htbp] 
	    \centering
	    \includegraphics[width=0.4\columnwidth]{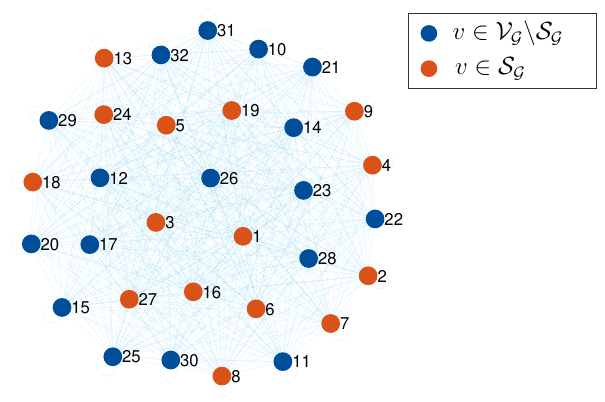}
	    \caption{The schematic diagram of $\S_\G$ for EEG ($B_\G = 15$). } 
	\label{fig:exp2_G_SG}
    \end{figure}

    For different values of $B_\G$ (1 to 32), the same DTVGS will be compressed into different JBL DTVGS. Therefore, we finally sampled and reconstructed 3200 JBL DTVGS using the multi-band sampling scheme. In \cref{fig:exp2_NRMSE_EEG} (a), we show the average sampling ratio for 100 JBL DTVGS corresponding to each $B_\G$ value.
    
    The normalized Root Mean Square Error (NRMSE) is used to describe the error between signals $\X_a$ and $\X_b$:
\begin{equation*}
    \text{NRMSE}(\X_a,\X_b) = \frac{ \| \X_a-\X_b \|_2 }{ \| \X_a \|_2 }.
\end{equation*}
    For each JBL DTVGS, we calculate $\text{NRMSE}(\X_{\rm JBL},\hat{\X})$ and $\text{NRMSE}(\X,\hat{\X})$ when $B_\G$ takes different values. The average of the NRMSE of the 100 JBL DTVGS corresponding to each $B_\G$ value are shown in \cref{fig:exp2_NRMSE} (b) and (c).
    
    \begin{figure}[htbp]
    \centering    
    \subfigure[]
    {
	    \includegraphics[scale=0.36]{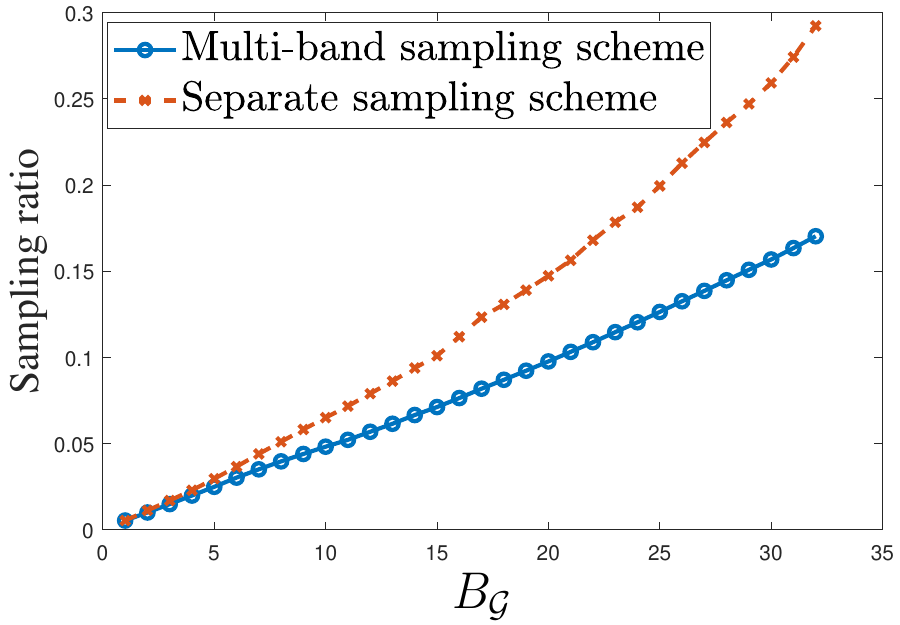}   
	}
    \subfigure[]
    {
	    \includegraphics[scale=0.36]{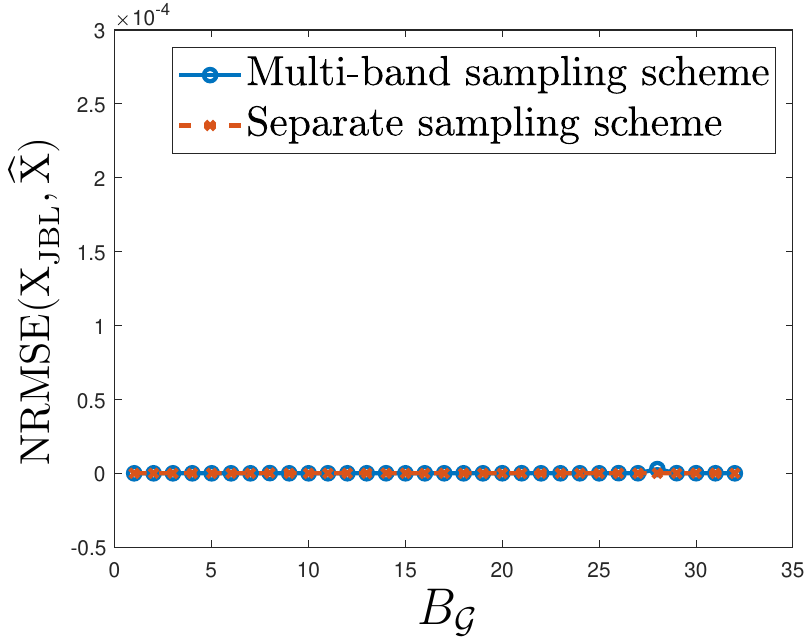}   
    }
    \subfigure[]
    {
	    \includegraphics[scale=0.36]{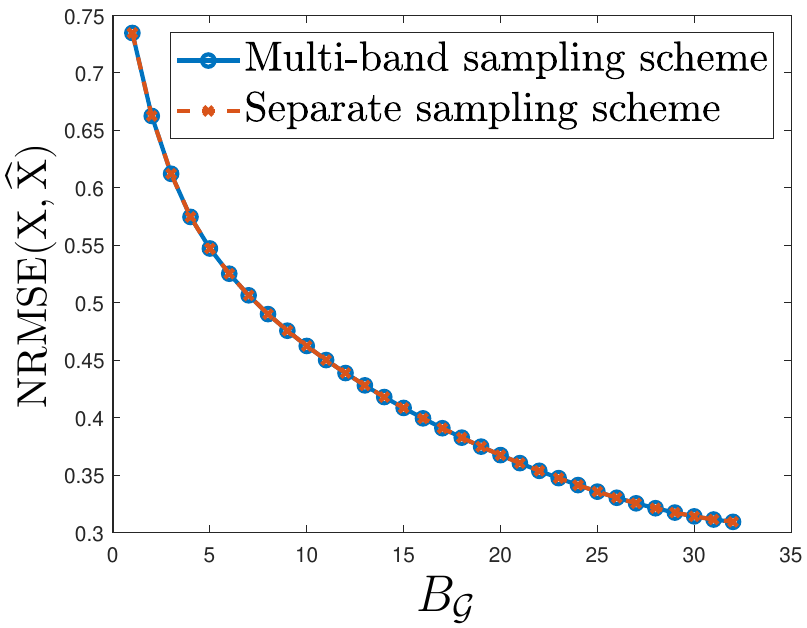}   
    }
    \caption{The averaged sampling ratios and NRMSE of EEG. (a) The average of the sampling ratio with different $B_\G$. (b) and (c) are the averages of $\text{NRMSE}(\X_{\rm JBL},\hat{\X})$ and $\text{NRMSE}(\X,\hat{\X})$ with different $B_\G$, respectively.}
    \label{fig:exp2_NRMSE_EEG} 
    \end{figure}

\subsubsection{Test on METR-LA}
    
    Graph $\G$ of the EEG was constructed based on correlation. To make the experiment more general, we also tested the multi-band sampling scheme on a dataset METR-LA that provides the graph structure and the signals. Similar to the EEG data, we divided the METR-LA dataset into 100 DTVGS denoted as $\X$, where each DTVGS consists of $N=207$ vertices, and each vertex relates to a discrete sequence of 1024 timesteps.
    
    The sensor distribution of METR-LA is visualized in \cref{fig:exp2_geo_SG}, and the topology of the dataset is given in \cite{data_traffic} (adjacency matrix $\mathbf{A}$). Since the data is modeled as a directed graph in vertex domain in \cite{data_traffic}, resulting in an asymmetric adjacency matrix $\mathbf{A}$. To ensure consistency with the model in this paper, we convert the directed graph to an undirected graph by letting $\W_\G = (\mathbf{A}+\mathbf{A}^H)/2$ \cite{direc_A}, as shown in \cref{fig:exp2_W}. 
    
    Once again, the theories presented in this paper are applicable even when modeling the graph domain topology of the three kinds of TVGS as directed graphs.
    
    \begin{figure} [htbp] 
	    \centering
	    \includegraphics[width=0.3\columnwidth]{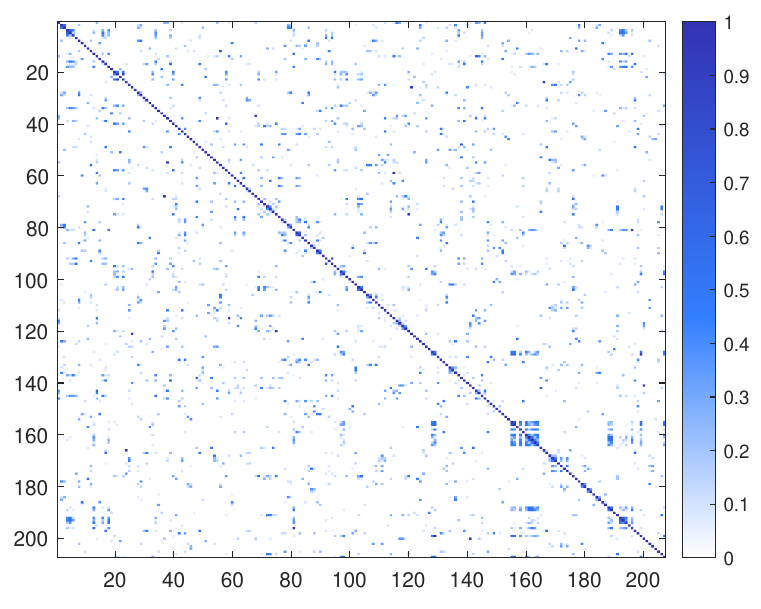}
	    \caption{Weighted adjacency matrix $\W_\G$ of METR-LA. } 
	\label{fig:exp2_W}
    \end{figure}
    
    For a DTVGS $\X$, we show the energy of $\X$, $\F_{DT}(\X)$, $\F_\G(\X)$, and $\F_\J(\X)$ in \cref{fig:exp2_signal}. To obtain strictly bandlimited signals, we apply a low-pass filter to each row of $\F_\J(\X)$ and reserve the $B_\G$ rows, following a similar procedure as with the EEG data. The resulting JBL DTVGS is denoted as $\X_{\rm JBL}$. Taking the signal in \cref{fig:exp2_signal} as an example, with $B_\G=100$, $\X_{\rm JBL}$ retains $97.27\%$ of the energy in $\X$, with bandwidths of $B=0.0181$ and $B_\T=0.0008$.
    
    \begin{figure} [htbp] 
	    \centering
	    \includegraphics[width=0.7\columnwidth]{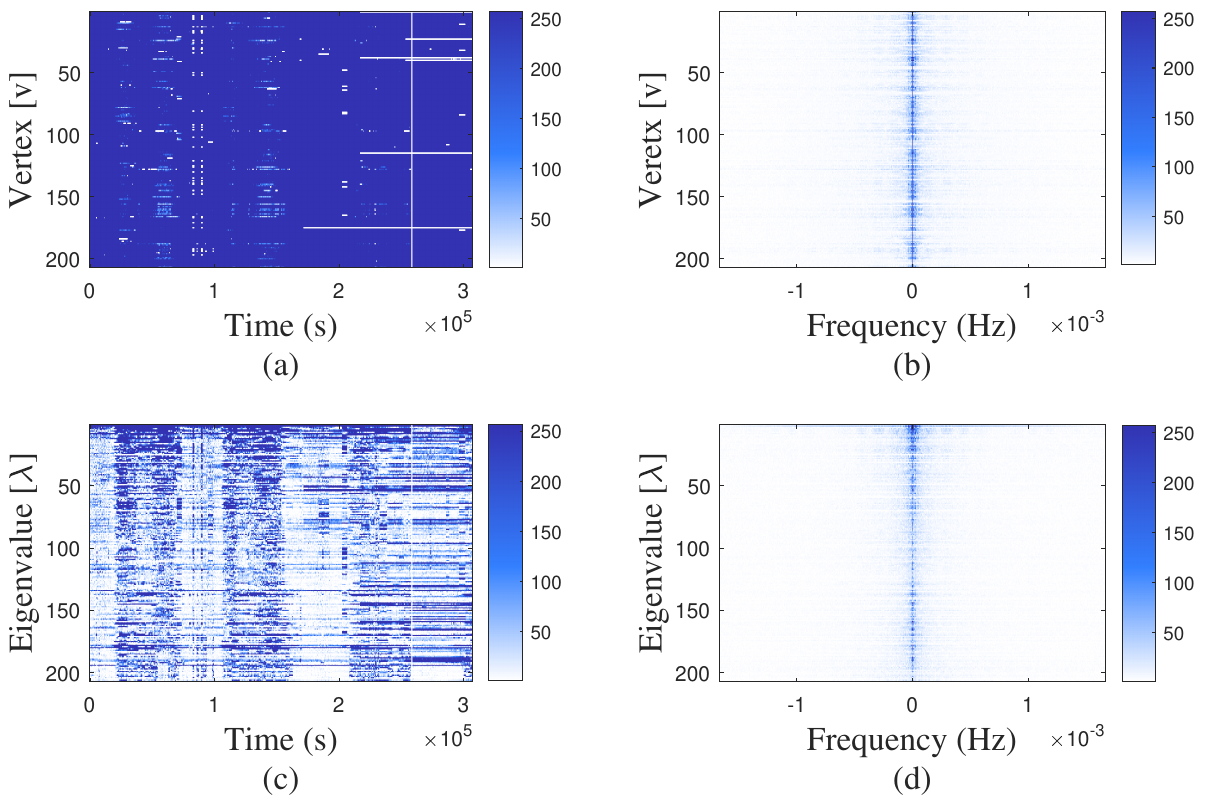}
	    \caption{Energy analysis of $\X$ getting from METR-LA. (a) is the energy of the original signal $\X$, whose each row is a sequence on a vertex, and each column is a graph signal at a certain instant. (b), (c), and (d) are the energy of $\F_{DT}(\X)$, $\F_\G(\X)$, and $\F_\J(\X)$, respectively.}
	\label{fig:exp2_signal}
    \end{figure}
    
    \emph{Sampling and reconstruction}: We sample and reconstruct each $\X_{\rm JBL}$ with the multi-band sampling scheme, and the corresponding recovered signal is recorded as $\hat{\X}$. The $\X_{\rm JBL}$ described above can be sampled with a total sampling ratio of $R_D(\S) = B/N = 0.0262$, less than $0.1227$ of separate sampling. The set of sampled vertices $\S_\G$ is shown in \cref{fig:exp2_geo_SG}. We exactly get the critical sampling set. 
    
    \begin{figure} [htbp] 
	    \centering
	    \includegraphics[width=0.6\columnwidth]{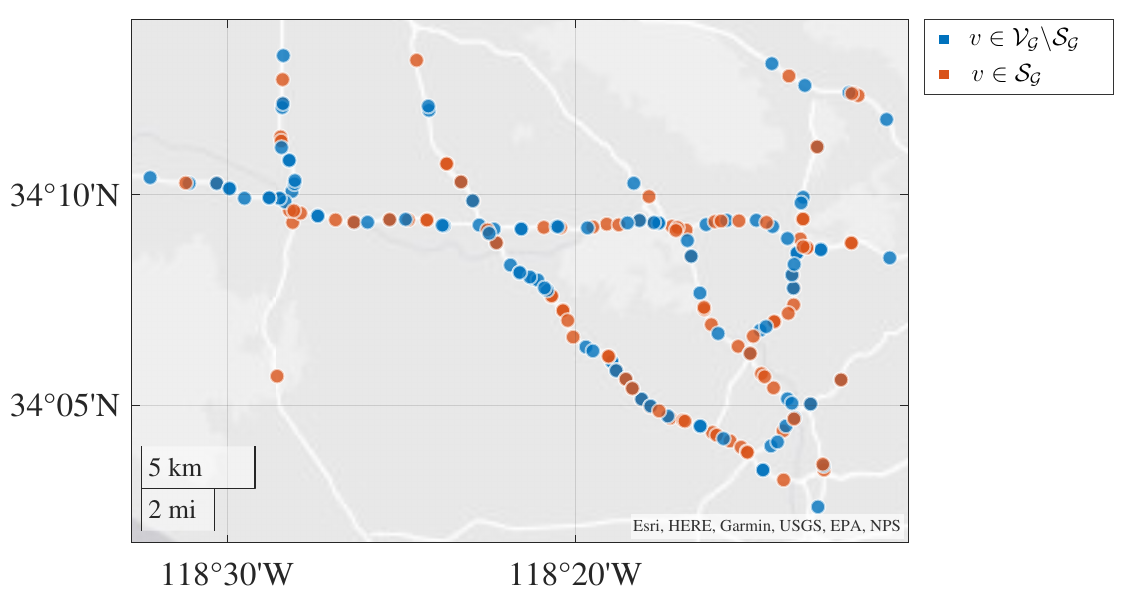}
	    \caption{The schematic diagram of $\S_\G$ for METR-LA ($B_\G = 100$). } 
	\label{fig:exp2_geo_SG}
    \end{figure}

    By varying the value of $B_\G$, we sample and reconstruct a total of 20,700 JBL DTVGS. We calculate the average values of the sampling ratios for the 100 $\X_{\rm JBL}$ with the same $B_\G$, as well as the $\text{NRMSE}(\X_{\rm JBL},\hat{\X})$ and $\text{NRMSE}(\X,\hat{\X})$. These results are presented in \cref{fig:exp2_NRMSE}.
    
    \begin{figure}[htbp]
    \centering    
    \subfigure[]
    {
	    \includegraphics[scale=0.36]{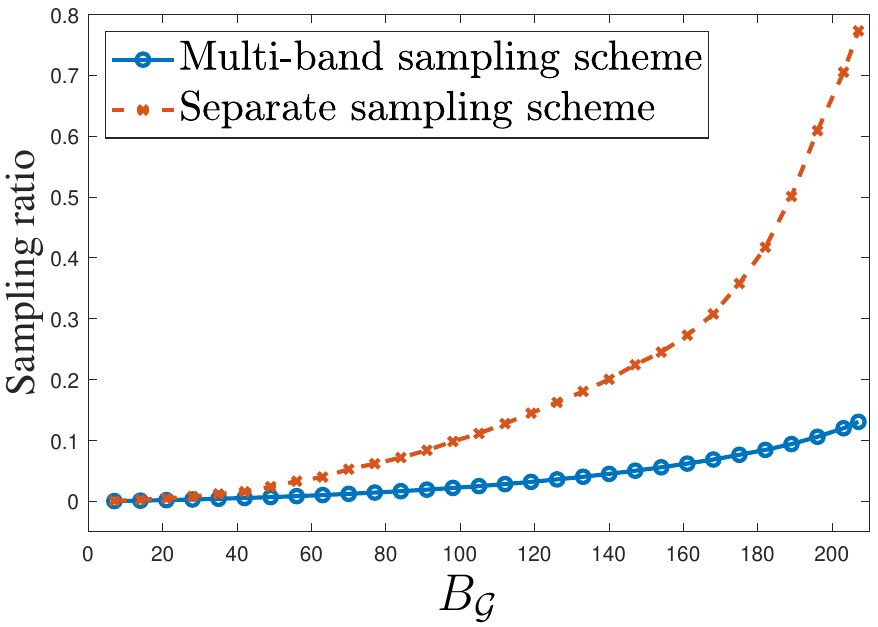}   
	}
    \subfigure[]
    {
	    \includegraphics[scale=0.36]{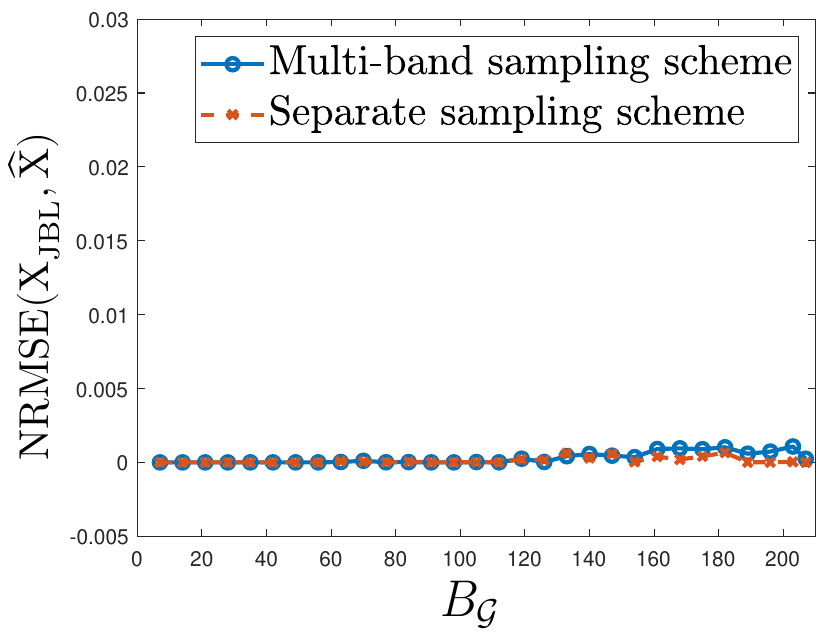}   
    }
    \subfigure[]
    {
	    \includegraphics[scale=0.36]{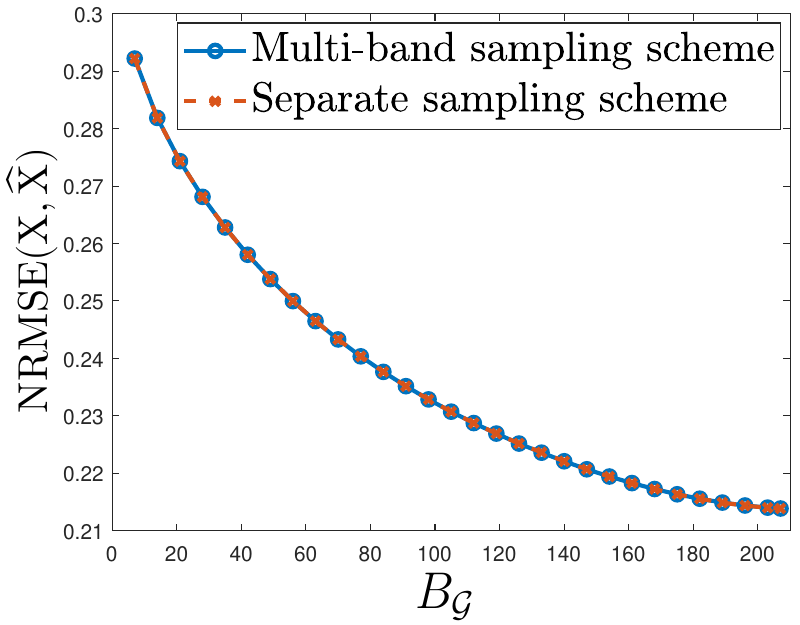}   
    }
    \caption{The averaged sampling ratios and NRMSE of METR-LA. (a) The average of the sampling ratio with different $B_\G$. (b) and (c) are the averages of $\text{NRMSE}(\X_{\rm JBL},\hat{\X})$ and $\text{NRMSE}(\X,\hat{\X})$ with different $B_\G$, respectively.}
    \label{fig:exp2_NRMSE} 
    \end{figure}

\subsubsection{Discussion}
    
    \emph{Spectrum analysis}: 
    The ability of GFT to encode graph-dependent signals compactly is the motivation behind the sampling based on $\F_\J(\X)$. In \cref{fig:exp2_signal_EEG} and \cref{fig:exp2_signal}, we analyze the energy of $\X$ and its spectra for both the EEG and METR-LA. We can easily see that $\X$ does not emphasize the relation between the time domain and vertex domain. Compared with $\X$, FFT and GFT compact the energy distribution along the rows and columns, respectively. \emph{Remarkably, $\F_\J(\X)$ successfully represents the signal in a more efficient way, as it exhibits greater sparsity compared to $\X$, $\F_\G(\X)$, and $\F_{FT}(\X)$. }

    \emph{Sampling ratio}: 
    For both the EEG signal and METR-LA signal, when $B_\G = 1$, we only need to sample the signal on a single vertex, equivalent to the sampling problem in classical signal processing. The sampling ratios are the same for both methods. With the increase of $B_\G$ and $B$, the sampling ratios of both two sampling schemes also increase, shown in \cref{fig:exp2_NRMSE_EEG} (a) and \cref{fig:exp2_NRMSE} (a). \emph{No matter how much $B_\G$ is, the sampling ratio of our method is no more than that of separate sampling.}

    \emph{NRMSE}: 
    Testing on EEG or METR-LA, when $B_\G$ takes different values, $\text{NRMSE}(\X_{\rm JBL},\hat{\X})$ obtained by the multi-band sampling scheme and the separate sampling scheme are both close to zero (\cref{fig:exp2_NRMSE_EEG} (b) and \cref{fig:exp2_NRMSE} (b)), which shows the robustness of our multi-band sampling scheme. The $\text{NRMSE}(\X,\hat{\X})$ is mainly derived from the operation making the signal strictly bandlimited. As $B_\G$ increases, $\X_{\rm JBL}$ is getting closer to $\X$, so $\text{NRMSE}(\X,\hat{\X})$ gradually decreases (\cref{fig:exp2_NRMSE_EEG} (c) and \cref{fig:exp2_NRMSE} (c)).

\section{Conclusion} \label{sec:conclusion}

In this work, based on the time-vertex signal processing framework, we propose and prove the necessary conditions for stable sampling and reconstruction for three kinds of TVGS: CTVGS, DTVGS, and FTVGS. 

We use ideal filters to cut subbands in the proof for the sake of clearness in theory, which can not be used in practical projects. It is possible to replace the ideal filters with other filter banks in practice, allowing a certain level of reconstruction error. In addition, we assume that all the vertices are sampled synchronously. We will consider how to reconstruct the asynchronous sampling scheme in the follow-up research.

\appendix
\section{Proof of the Theorem \ref{thm:subset_f}}
\label{pf_th_f}

\begin{proof}
    Let $\mathcal{N}$ be the set of index of nonzero elements of $\F_\J(\x)$. Here, when we consider the ratio of $\S_\Theta$, we assume that the signal in $\Theta^c$ is known. So we have 
    \begin{equation*}
        \left[ \begin{matrix} \mathbf{\Psi}_\theta \\ \textbf{0} \end{matrix} \right] \U_\J(\S_\G' \times \V_\T, \mathcal{N}) \F_\J(\x)(\mathcal{N}) = \x(\S_\G' \times \V_\T) - \left[ \begin{matrix} \textbf{0} \\ \mathbf{\Psi}_\theta^c \end{matrix} \right] \U_\J(\S_\G' \times \V_\T, \mathcal{N}) \F_\J(\x)(\mathcal{N}),
    \end{equation*}
    where $\mathbf{\Psi}_\theta \in \{0, 1\}^{|\Theta|T \times B_\G T}$ is the sampling matrix that selects elements in $\{ \Theta \times \V_\T \}$ from $\x(\S_\G' \times \V_\T)$. Matrix $\left[ \begin{matrix} \mathbf{\Psi}_\theta \\ \mathbf{\Psi}_\theta^c \end{matrix} \right] \in \{0, 1\}^{B_\G T \times B_\G T}$ is obtained from the row transformation of an identity matrix. Then we get the following low-dimensional representation
    \begin{equation}
    \label{eq:IJFT_lowD_f}
        \x(\S_\G' \times \V_\T) = \U_\J(\S_\G' \times \V_\T, \mathcal{N}) \F_\J(\x)(\mathcal{N}).
    \end{equation}
    
    According to Theorem 1 in \cite{theory}, we have $\mathbf{\Psi} \x(\S_\G' \times \V_\T) = \mathbf{\Psi} \U_\J(\S_\G' \times \V_\T, \mathcal{N}) \F_\J(\x)(\mathcal{N})$, and the number of sampled elements on $\x(\S_\G' \times \V_\T)$ cannot be less than ${\rm rank}(\U_\J(\S_\G' \times \V_\T, \mathcal{N}))$. 
    
    Since ${\rm rank} (\U_\G(\S_\G', \I)) \ge B_\G$,  $\U_\J(\S_\G' \times \V_\T, \I \times \V_\T) = \U_\G(\S_\G', \I) \otimes \U_\T$ is a full column rank matrix. The column set of $\U_\J(\S_\G' \times \V_\T, \mathcal{N})$ is a subset of $\U_\J(\S_\G' \times \V_\T, \I \times \mathcal{F})$. So we have
    \begin{equation*}
        {\rm rank}(\U_\J(\S_\G' \times \V_\T, \mathcal{N})) = {\rm rank}(\U_\J(\S_\G' \times \V_\T, \cdot) \mathbf{\Psi}_j^H) = B.
    \end{equation*}
    In addition, there must be a matrix $\mathbf{\Psi}$ with $|\mathbf{\Psi}| = B$ that makes $\mathbf{\Psi} \U_\J(\S_\G' \times \V_\T, \mathcal{N})$ full rank, \emph{i.e.}, the minimum singular value of matrix $\mathbf{\Psi} \U_\J(\S_\G' \times \V_\T, \mathcal{N})$ is greater than zero. Thus the sampling operation is stable.
    
    On the premise of stable sampling, we next discuss the lower bound of the sampling ratio of $\S_\Theta$. We consider sampling on $ \x(\Theta \times \V_\T) = \mathbf{\Psi}_\theta \x(\S_\G' \times \V_\T)$. There must be
    \begin{equation}
    \label{eq:card}
        |\S_\Theta| \ge {\rm rank}( \mathbf{\Psi}_\theta \U_\J(\S_\G' \times \V_\T, \cdot) \mathbf{\Psi}_j^H ),
    \end{equation}
    
    The rank property of the matrix makes the inequality hold: 
    \begin{equation}
    \label{eq:rank}
        {\rm rank}( \left[ \begin{matrix} \mathbf{\Psi}_\theta \\ \mathbf{\Psi}_\theta^c \end{matrix} \right] \U_\J(\S_\G' \times \V_\T, \cdot) \mathbf{\Psi}_j^H ) \le {\rm rank}(\mathbf{\Psi}_\theta \U_\J(\S_\G' \times \V_\T, \cdot) \mathbf{\Psi}_j^H) + {\rm rank}(\mathbf{\Psi}_\theta^c \U_\J(\S_\G' \times \V_\T, \cdot) \mathbf{\Psi}_j^H).
    \end{equation}
    Additionally, 
    \begin{equation*}
        {\rm rank}( \left[ \begin{matrix} \mathbf{\Psi}_\theta \\ \mathbf{\Psi}_\theta^c \end{matrix} \right] \U_\J(\S_\G' \times \V_\T, \cdot) \mathbf{\Psi}_j^H ) = B.
    \end{equation*}
    
    Combining \cref{eq:card}) and \cref{eq:rank}), we obtain
    \begin{equation*}
    \begin{aligned}
        |\S_\Theta| & \ge {\rm rank}(\mathbf{\Psi}_\theta \U_\J(\S_\G' \times \V_\T, \cdot) \mathbf{\Psi}_j^H)\\
        &\ge B - {\rm rank}(\mathbf{\Psi}_\theta^c \U_\J(\S_\G' \times \V_\T, \cdot) \mathbf{\Psi}_j^H).
    \end{aligned}
    \end{equation*}
    
    Additionally, $T \in \mathbb{N}_+$, thus
    \begin{equation*}
        R_F(\S) \! = \! \frac{|\S_\Theta|}{|\Theta|T} \! \ge \! \frac{1}{|\Theta|T} \! \left( B \! - \! {\rm rank}( \mathbf{\Psi}_\theta^c \U_\J(\S_\G' \times \V_\T, \cdot) \mathbf{\Psi}_j^H ) \right).
    \end{equation*}
\end{proof}

\bibliographystyle{unsrtnat}

\bibliography{references}

\end{document}